\documentclass[a4paper,reqno,11pt]{amsart}
\calclayout
\usepackage{amssymb}
\usepackage{color}
\usepackage[utf8]{inputenc}
\usepackage{tikz}
\usepackage[all]{xy}
\usepackage[bookmarksnumbered,colorlinks]{hyperref}
\usepackage{dsfont}
\usepackage[normalem]{ulem}

\usepackage[euler-digits]{eulervm}
\usepackage{enumerate}
\def\br#1\er{\textcolor{red}{#1}} %
\newcommand{\R}{\mathds R}
\newcommand{\N}{\mathds N}

\newcommand{\lev}[2]{\overline{#1}_{#2}}

\newcommand{\lcrono}{L_{chr}}

\title[Spacetimes coverings and C-boundary]{Spacetimes coverings and C-boundary}
\author[L. Ak\'e]{Luis Alberto Ak\'e}
\address{Departamento de \'Algebra, Geometr\'{\i}a y Topolog\'{\i}a,  Universidad de M\'alaga
\hfill\break\indent
Facultad de Ciencias, Campus Universitario de Teatinos,
\hfill\break\indent 29080 M\'alaga, Spain}
\email{luisake@uma.es}

\author[J. Herrera]{J\'onatan Herrera}
\address{Department of Mathematics,
	Universidade Federal de Santa Catarina,\hfill\break\indent 88.040-900 Florian\'{o}polis-SC, Brazil.}
\email{jonatanhf@gmail.es}

\begin{document}
\newtheorem{thm}{Theorem}[section]
\newtheorem{prop}[thm]{Proposition}
\newtheorem{lemma}[thm]{Lemma}
\newtheorem{cor}[thm]{Corollary}
\newtheorem{conv}[thm]{Convention}
\theoremstyle{definition}
\newtheorem{defi}[thm]{Definition}
\newtheorem{notation}[thm]{Notation}
\newtheorem{exe}[thm]{Example}
\newtheorem{conj}[thm]{Conjecture}
\newtheorem{prob}[thm]{Problem}
\newtheorem{rem}[thm]{Remark}
\maketitle
\usetikzlibrary{matrix}

\begin{abstract}

We consider the relation between the c-completion of a Lorentz manifold $V$ and its quotient $M=V/G$, where $G$ is an isometry group acting freely and properly discontinuously.
First, we consider the future causal completion case, characterizing virtually when such a quotient is well behaved with the future chronological topology and improving the existing results on the literature.  Secondly, we show that under some general assumptions, there exists an homeomorphism and chronological isomorphism between both, the c-completion of $M$ and some adequate quotient of the c-completion of $V$ defined by $G$. Our results are optimal, as we show in several examples. Finally, we give a practical application by considering isometric actions over Robertson-Walker spacetimes, including in particular the Anti-de Sitter model.  
\end{abstract}

\tableofcontents

\section{Introduction}

The AdS/CFT correspondence, also known as Maldacena's duality, states the duality between gravitational theories, as string or M-theory, on a bulk space (usually a product of the Anti-de Sitter spacetime with spheres or other compact sets) and conformal field theories defined on the boundary of the bulk space which behaves as an hologram of inferior dimension (see \cite{Mal}). As it is apparent, the conjecture relies strongly in the notion of boundary of Lorentz manifolds. However, the problem to attach a natural boundary for any Lorentz manifold encoding relevant information on it, as its conformal structure and related elements (event horizons, singularities, etc.) has been a long standing issue along the last four decades.

%
Among the several constructions proposed (see \cite{GpScqg05,H3,S} for nice reviews on the classical elements and \cite{FHSFinalDef,FHSHaus} for updated progress), two approaches have had a specially important role in general relativity, the conformal and the causal boundaries.

The conformal boundary is the most applied one in mathematical relativity and several notions, as asymptotic flatness or tools as Penrose--Carter diagrams rely on it. Even in the original approach of the AdS/CFT correspondence, it is the conformal boundary the chosen as the holographic one. In fact, the Anti-de Sitter spacetime can be conformally embedded in the Lorentz-Minkowski model, obtaining a simple (and non-compact) conformal boundary.
However, it has important limitations as it is an {\em ad hoc} construction: no general formalism determines when the boundary of a reasonably general spacetime is definable, intrinsic, unique and contain useful information of the spacetime (see \cite{Chrusciel} and \cite[Section 4]{FHSFinalDef} for studies regarding the uniqueness of the conformal boundary). In fact, as it was putted forward by Bernstein, Maldacena and Nastase \cite{BMN}), there seems to be problems when the conformal boundary is considered on plane waves. Indeed, Marolf and Ross \cite{MR1} realized that the conformal boundary is not available for non-conformally flat plane waves. So, they proposed a redefinition of the c-boundary applicable to such waves \cite{MR} which was refined and systematically studied by Flores and Sanchez in \cite{FS2}.

This motivated a reconsideration of such constructions by substituting the conformal boundary by the causal one, which is intrinsic, conformally invariant and it can be computed systematically, as it was carried out in \cite{FHSFinalDef}. It is worth emphasizing that both the conformal and causal boundaries are shown to coincide in most relevant classes (so, previous results based on the conformal case are not required to be re-obtained for the causal one).


%

Returning to the problem of AdS/CFT correspondence, it is our aim to present the causal boundary of different classes of Lorentz manifolds, allowing the study of such a correspondence with different bulk spaces. In this sense let $M$ be, for instance, a Lorentz manifold with constant negative curvature, and so, a spacetime that can be locally modelled by the Anti-de Sitter spacetime. Recalling that the universal covering $\tilde{AdS}$ is maximal, simply-connected and with constant negative curvature, it is expectable that $M$ can be described as a quotient space of $\tilde{AdS}$ by an appropriate group of isometries (in fact, for certain spacetime topologies, the existence of such an appropriate group was proven by Mess \cite{Me}). This is the particular case of the BTZ blackholes, the $(2+1)$-model of spacetime first introduced by Ba\~nados, Teitelboim and Zanelli \cite{BTZ}; and the Hawking-Page reference space \cite{HP},  whose representations as a quotient of the Anti-de Sitter model are well known \cite{BHTZ,Wit1,Wit2}.
	
Due the fact that the causal boundary is well known for $\tilde{AdS}$ (see \cite[Section 4.1]{AF}), the following question, particularly natural from the mathematical viewpoint, arises: given two (general) Lorentz manifolds $M$ and $V$ where $M$ is constructed as the quotient of $V$ by some group of isometries, what is the relation between the causal boundaries and completions of $M$ and $V$? An adequate answer for this question will give us tools to easily compute the causal completion of $M$ once we known the corresponding on $V$. For instance, such a result will be applicable to models like the BTZ blackholes or the Hawking-Page reference model, besides other models constructed in a similar way (as the case of Cosmic Strings, see \cite{Got}). It will also give us relevant information of the c-completion on $V$ whenever the c-completion in $M$ is known. 
	
	The first studies in this direction are due to Harris \cite{H}. In his work, he studied how isometrical actions affect the causal structures of the spacetimes, with special attention to the future causal boundary and related concepts (as strong causality).  Concretely, he considers a projection $\pi:V\rightarrow M$ given by a discrete subgroup $G$ of isometries acting freely and properly discontinuously in $V$, i.e., where $M=V/G$ and the elements on $M$ represents $G$-orbits in $V$. In this settings, Harris characterizes the strong causality and global hyperbolicity of $M$ in terms of the global causal structure of $V$. Moreover, and under the assumption of $M$ being distinguishing (which implies, in particular, that $V$ also is), he presents necessary conditions in order to ensure when the associated quotient $\hat{V}/\hat{G}$ is homeomorphic to $\hat{M}$.

Our aim in this work is to extend the results obtained by Harris for the future causal completion to the c-completion. However, several problems have to be addressed first. On the one hand, the main result in \cite{H} imposes that both, the causal boundary of $M$ and $V$ have only spacelike future boundaries. This condition, even if reasonable (specially recalling the final example of his paper), is too strong for the c-completion context, where particularly timelike boudary points are specially relevant. On the other hand, and in spite with the partial case, the c-completion requires the study of the so-called S-relation between future and past sets, as well as some ``compatibility'' between the topology of the future and past completions.

\smallskip 

The contents of the paper are organized as follow. In Section \ref{sec:preliminares} we will give the preliminaries that we are going to need along the rest of the paper. Most of them are well known (for instance, the construction of the c-completion was developed in \cite{FHSFinalDef}), but we have also introduced concepts (as first order UTS, Definition \ref{def:UTS}) and results (Lemmas \ref{lem:lemclosed} and \ref{lem:firstUTS}; and some of the assertions in Theorem \ref{thm:causalladder}) that, as far as we know, are new.

Section \ref{sec:Partial} is devoted to the study of the future (and, by analogy, past) causal boundary. Here, at the point set level, we will recall the bijection $\hat{\jmath}$ defined by Harris between $\hat{V}/\hat{G}$ (two points in $\hat{V}$ are related if they project onto the same point in $\hat{M}$) and $\hat{M}$. Then, we will perform a detailed comparison between the topologies in both spaces (the first one with the induced quotient topology). The results of this section are summarized as follow:

\begin{thm}\label{thm:main1}
Let $\pi:V\rightarrow M$ be an spacetime covering projection (see Section \ref{sec:spaceproy}) and denote by $\hat{\pi}$ the extension to future c-completions \eqref{eq:aux4}. Let $\hat{V}/\hat{G}$ be the quotient space defined by the following relation: two points $P,P'\in \hat{V}$ are $\sim_{\hat{G}}$-related if they project onto the same point in $\hat{M}$. Then, we obtain the following commutative diagram:

\begin{center} 
	\begin{tikzpicture}
	\matrix (m) [matrix of math nodes,row sep=3em,column sep=4em,minimum width=2em]
	{
		\hat{V} & \\
		\hat{V}/\hat{G} & \hat{M} \\};
	\path[-stealth]
	(m-1-1) edge node [left] {$\hat{\i}$} (m-2-1)
	(m-2-1.east|-m-2-2) edge node [below] {$\hat{\j}$} (m-2-2)
	(m-1-1) edge node [above] {$\hat{\pi}$} (m-2-2);
	\end{tikzpicture}
\end{center}
where $\hat{\i}$ is the natural quotient projection. From construction, the map $\hat{\jmath}$ is bijective. At the topological level,  

\begin{itemize}
	\item[(i)] The map $\hat{\jmath}$ is open.
	\item[(ii)] If $M$ does not admit sequences with future divergent lifts (Definition \ref{def:divlif}), the map $\hat{\pi}$ (and so, $\hat{\jmath}$) is continuous. The converse also follows if we have that $\hat{L}_{M}$ is of first order UTS (Definition \ref{def:UTS}).
\end{itemize}

In particular, if $M$ has only spatial future boundary points, $\hat{\j}$ is an homeomorphism between $\hat{M}$ and $\hat{V}/\hat{G}$. The same result follows if $G$ is finite and $\hat{V}$ is Hausdorff.
\end{thm}
As we can see on previous (ii), we have obtained almost a characterization of the continuity of $\hat{\jmath}$, up to the first order UTS property. In fact, such a result generalizes \cite[Theorem 3.4]{H}, as the last assertion of Theorem \ref{thm:main1} shows.

Section \ref{sec:total} is focused on the study of the (total) c-completion at all possible levels, namely, at the point set, at the chronological and at the topological level. In Section \ref{sec:totalpoint} are given simple and general sufficient conditions to ensure the definition of the map $\overline{\j}$ between a reasonable quotient of $\overline{V}$ by $G$ (see Definition \ref{def:proycompleto}) and $\overline{M}$. Then, it is shown in Section \ref{sec:chronology} that previous map is well behaved respect the chronological relation, whenever an appropriate chronological relation is defined on the quotient space $\overline{V}/G$. Then, in Sections \ref{sec:totaltopo} and \ref{sec:totalfinite}, it is studied the conditions to ensure that the map $\overline{\jmath}$ is both, continuous and open resp. Now the latter becomes subtler and a simple condition (to be {\em finitely chronological}) is introduced. This property also simplify the conditions to ensure the well posedness and continuity of $\overline{\jmath}$. 

Concretely, the results of such a section are summarized in the following:

\begin{thm}\label{thm:main2}
Let $\pi:V\rightarrow M$ be a spacetime covering projection and consider $\overline{\pi}$ the extension map to the corresponding c-completions as defined on Definition \ref{def:proycompleto}. Then:

\begin{itemize}
	\item[(PS1)] The projection is well defined and surjective if $M$ does not admit sequences with (future or past) divergent lifts and any $(P,F)\in \overline{M}$ with $P\neq\emptyset\neq F$ admits a lift on $\overline{V}$. In particular, the latter condition holds if $(V,G)$ is finitely chronological (Definition \ref{def:finitelyachronal}).
	
	\item[(PS2)] If, in addition, the projection $\pi$ is tame (recall Definition \ref{rem:tameextension}), $\overline{\pi}$ just reads as
	
	\[
	\overline{\pi}((\lev{P}{},\lev{F}{}))=(\hat{\pi}(\lev{P}{}),\check{\pi}(\lev{F}{}))
	\]
	 
	 \item[(PS3)] The projection $\overline{\pi}$ is well defined if $(V,G)$ is finitely chronological and $\overline{V}$ is Hausdorff.
	 
\end{itemize}

\medskip

\noindent Moreover, when the map $\overline{\pi}$ is well defined and surjective, it defines the following relation between points in $\overline{V}$: two points are $\sim_{G}$-related if they projects onto the same point in $\overline{M}$. Then, denoting by $\overline{V}/G$ the quotient space, we obtain the following commutative diagram:

\begin{center} 
	\begin{tikzpicture}
	\matrix (m) [matrix of math nodes,row sep=3em,column sep=4em,minimum width=2em]
	{
		\overline{V} & \\
		\overline{V}/G & \overline{M} \\};
	\path[-stealth]
	(m-1-1) edge node [left] {$\overline{\i}$} (m-2-1)
	(m-2-1.east|-m-2-2) edge node [below] {$\overline{\j}$} (m-2-2)
	(m-1-1) edge node [above] {$\overline{\pi}$} (m-2-2);
	\end{tikzpicture}
\end{center}
where $\overline{\i}$ is the natural projection to the quotient and $\overline{\jmath}$ is the induced bijection. 

\smallskip 

{\em At the chronological level}, and once an appropriate chronological relation is defined on $\overline{V}/G$ (see Section \ref{sec:chronology}), it follows that

\begin{itemize}
	\item[(CH)] the map $\overline{\jmath}$ is a chronological isomorphism.
\end{itemize}

\smallskip 

 {\em Finally, at the topological level}, $\overline{\jmath}$ satisfies the following properties: 

\begin{itemize}
	\item[(TP1)] The map $\overline{\jmath}$ is continuous if one of the following hypotheses hold:
	\begin{itemize}
		\item[(i)] $\overline{\pi}$ satisfies that  $\overline{\pi}(\lev{P}{},\emptyset)=(P,\emptyset)$ and $\overline{\pi}(\emptyset,\lev{F}{})=(\emptyset,F)$ (this follows if, for instance, $\pi$ is a tame projection); and $M$ has no sequence with (future or past) divergence lifts
		\item[(ii)] $\overline{\pi}(\lev{P}{},\emptyset)=(P,\emptyset),$ $\overline{\pi}(\emptyset,\lev{F}{})=(\emptyset,F)$  and $\overline{M}$ has no lightlike boundary points\footnote{Here, we say that the boundary has no lightlike boundary points if given a  point $(\lev{P}{},\emptyset)\in \overline{V}$ (resp. $(\emptyset,\lev{F}{})\in \overline{V}$) there is no indecomposable past set $\lev{P'}{}$ (future set $\lev{F'}{}$) such that $\lev{P}{}\subset \lev{P'}{}$ ($\lev{F}{}\subset \lev{F'}{}$).}.
	\end{itemize}
	 
	\item[(TP2)] If $(V,G)$ is finite chronological, the map $\overline{\jmath}$ is open. 
\end{itemize}

\smallskip 

In particular, $\overline{\pi}$ is well defined, surjective and induces an homeomorphism and chronological isomorphism between $\overline{V}/G$ and $\overline{M}$ if it is satisfied one of the following assertions:
\begin{itemize}
	\item[(a)] $\pi$ is tame, $(V,G)$ is finite chronological and $M$ admits no sequence with (future or past) divergent lifts.
	
	\item[(b)] $(V,G)$ is finitely chronological, $\overline{V}$ is Hausdorff and $\overline{M}$ has no lightlike boundary points.
	
	\item[(c)] $(V,G)$ is finitely chronological, $\overline{V}$ is Hausdorff, it has no lightlike boundary points and the $G$-orbits for both $\hat{V}$ and $\check{V}$ are closed.
\end{itemize}

\end{thm}

In Section \ref{sec:examples} we include several technical examples showing the optimality of our results, that is, we show that if we remove any of the three sufficient conditions (tameness, no existence of sequences with divergent lifts or finite chronology), the results are, in general, false. Finally in Section \ref{sec:application}, and as a physically relevant application of our result, we use the developed theory to compute the causal boundary of quotients of Robertson Walker spacetimes, including quotients of the AdS Spacetime.


\section{Preliminaries}\label{sec:preliminares}

\subsection{Sequential topologies and limit operators}
\label{sec:prellimits}
Along this section we will include all the basic facts about sequential topologies and limit operators that we will require for the rest of the paper. Most of the results are known (see \cite{FHSIso2,FHSHaus}), but we present the concept of first order UTS along some associated results that, as far as we known, are new.
\smallskip 

Let $X$ be an arbitrary space with a limit operator $L$ defined on it, that is, an operator $L:\mathcal{S}(X)\rightarrow \mathcal{P}(X)$, where $\mathcal{S}(X)$ is the space of sequences in $X$ and $\mathcal{P}(X)$ is the space of parts of $X$. We will always assume that the limit operator is {\em coherent}, i.e., that $L(\sigma)\subset L(\kappa)$ where $\kappa,\sigma\in \mathcal{S}(X)$ and $\kappa$ is a subsequence of $\sigma$ (this will be denoted by $\kappa\subset \sigma$).

Any coherent limit operator defines naturally a sequential topology $\tau_{L}$ on $X$ on the following way: a set $C$ is closed for $\tau_{L}$ if and only if $L(\sigma)\subset C$ for all sequence $\sigma\subset C$. Reciprocally, any sequential topology $\tau$ has associated a limit operator $L_{\tau}$ (its usual convergence) such that $\tau=\tau_{\tau_{L}}$ (see \cite[Proposition 2.6]{FHSHaus}).

In general, the limit operator $L$ does not determine the complete set of convergence points of a sequence $\sigma$ with the topology $\tau_{L}$. In fact, the only implication which is always true is that: 

\begin{equation}\label{eq:1}  
p\in L(\sigma) \Longrightarrow \hbox{$\sigma$ converges to $p$ with the topology $\tau_{L}$}.
\end{equation}
When the other implication is satisfied for all sequences, we will say that the limit operator is {\em of first order}. In general, there are not many results determining when a limit operator is of first order. In fact, in practical cases, the proof is done case by case, taking special care of ``problematic'' sequences. However, if we relax slightly the first order condition on $L$, we can obtain simply-to-check conditions which will be enough for our purposes. In this sense, let us introduce some definitions.

\begin{defi}\label{def:UTS}
	Let $X$ be a space and $L$ a limit operator defined on $X$. Let us denote by $\tau_{L}$ the associated sequential topology and let $\sigma\subset X$ be a sequence. We will say that $L$ is {\em of first order for $\sigma$} if
	
	\[
	p\in L(\sigma) \iff \hbox{$\sigma$ converges to $p$ with the topology $\tau_{L}$}.
	\] 
	Additionally, we will say that $L$ is {\em of first order up to a subsequence for $\sigma$} (or {\em first order UTS} for short), if $\sigma$ has a subsequence $\kappa\subset \sigma$ such that $\kappa$ is of first order for $L$. Finally, we will say that $L$ is of first order UTS if it is of first order UTS for all sequence $\sigma\subset X$.	
\end{defi}

The following result give us a sufficient condition to ensure when a limit operator is of first order for a given sequence.

\begin{lemma}\label{lem:lemclosed} 
	Let $\sigma$ be a sequence on $X$ such that, for all $\kappa\subset \sigma$, $L(\kappa)=L(\sigma)$. Assume additionally that $L(\sigma)$ only contains a finite number of elements. Then, $cl(\sigma)=\sigma\cup L(\sigma)$, where $cl(\sigma)$ denotes the topological closure of $\sigma$. In particular, $L$ is of first order for $\sigma$. 
\end{lemma} 
\begin{proof}
	The proof is quite straightforward and we include it here for the sake of completeness. Observe that the set $C=\sigma\cup L(\sigma)\subset cl(\sigma)$ from \eqref{eq:1}, so the first assertion follows if we prove that $C$ is closed. For this, let $\kappa\subset C$ and let us prove that $L(\kappa)\subset C$. Recall that, due the finite number of elements in $L(\sigma)$, we have have two possibilities (up to a subsequence) for $\kappa\subset C$: Or the sequence $\kappa$ is a subsequence of $\sigma$, and so, $L(\kappa)=L(\sigma)\subset C$; or $\kappa$ is constantly an element $p\in C$, and so, $L(\kappa)=\{p\}_{n} \subset C$. In both cases, $L(\kappa)\subset C$ and hence $C$ is closed. 
	
	For the last assertion, that is, the first order character of $L$ on $\sigma$, let us assume that $\sigma\rightarrow p$. Again, we distinguish two cases:
	
	\begin{itemize}
		\item We can exclude a finite number of elements in $\sigma$ such that the refined sequence $\sigma'$ does not contain $p$. As we are removing only a finite number of elements, $L(\sigma')=L(\sigma)$ and it follows from the first assertion that $cl(\sigma')=\sigma'\cup L(\sigma')$. As $\sigma'\rightarrow p$, we have that $p\in \sigma'\cup L(\sigma')$. From construction $\sigma'$ does not contain $p$, so $p\in L(\sigma')=L(\sigma)$.
		
		\item Otherwise, we can construct a subsequence $\kappa$ of $\sigma$ with $\kappa=\{p\}_n$. In particular, $p\in L(\{p\}_{n})=L(\kappa)=L(\sigma)$ (recall that the last equality follows by hypothesis). 
	\end{itemize}
	In conclusion, $p\in L(\sigma)$ and $L$ is of first order for $\sigma$. 
\end{proof}
%

Previous result give us a relatively simple way to determine when $L$ is of first order for a given sequence $\sigma$ and it is usually enough in particular cases. However, we can go a step further on the search of a easily verifiable condition. For this, let us note that most of the results we will present on this paper require, not a complete control of the convergence of sequences, but the existence for any sequence of a subsequence sufficiently well behaved. This is make apparent in the following result which ensure continuity of a map between sequential spaces:

\begin{prop}\label{prop:limcont}
	Let $f:(M,L)\rightarrow (N,L')$ be a map between sequential spaces $(M,L)$ and $(N,L')$. The map $f$ is continuous if for any sequence $\{p_n\}_{n}  \subset M$ and $p\in L(\{p_n\}_{n})$ there exists a subsequence $\{p_{n_k}\}_{k}$ such that $f(p)\in L'(\{f(p_{n_k})\}_{k})$.
\end{prop}

\begin{proof}
	Let $C$ be a closed set in $(N,L')$, and let us show that $f^{-1}(C)$ is closed on $(M,L)$. Assume by contradiction that $f^{-1}(C)$ is not closed and so, from definition, that there exists a sequence $\sigma\subset f^{-1}(C)$ and a point $p\in M$ with $p\in L(\sigma)\setminus f^{-1}(C)$. From hypothesis, there exists a subsequence $\kappa\subset \sigma$ such that $f(p)\in L'(f(\kappa))$. But $f(\kappa)\subset C$, which is closed for the topology $\tau_{L'}$. Therefore $f(p)\in C$, and so, $p\in f^{-1}(C)$, a contradiction.  
\end{proof}

This is one of the reasons why the condition of $L$ being of first order UTS is specially interesting for us. Moreover, as we can see on the following result, it is possible to obtain the following sufficient conditions for the first order UTS, which is particularly simple to verify in practical cases:

\begin{lemma}\label{lem:firstUTS}
	Let $X$ be any space with a limit operator $L$ defined on it. Assume that there exists $0\leq K<\infty$ such that $\#L(\sigma)\leq K$ for all sequence $\sigma\subset X$, where $\#L(\sigma)$ denotes the cardinality of the set $L(\sigma)$. Then, $L$ is of first order UTS.
\end{lemma}

\begin{proof}
	The proof is quite straightforward from the hypothesis and the fact that $L(\sigma)\subset L(\kappa)$ for all $\kappa\subset \sigma$. In fact, observe that for any sequence $\sigma$ one of the following possibilities appear:
	\begin{itemize}
		\item Or for all subsequence $\kappa\subset \sigma$, $L(\kappa)=L(\sigma)$.
		
		\item Or there exists $\kappa\subset \sigma$ with $L(\sigma)\subsetneq L(\kappa)$. In particular, $\#L(\kappa)\geq \#L(\sigma)+1$.
	\end{itemize}
	
	On the first case,  the sequence $\sigma$ is of first order according to Lemma \ref{lem:lemclosed} and we are done. On the second case, we repeat the same argument with $\kappa$ on the role of $\sigma$. Due to the fact that $\#L(\overline{\sigma})\leq K$ for any sequence $\overline{\sigma}$, previous process should end on a finite number of steps with a subsequence $\tilde{\kappa}$ on the first case, and the result follows.
\end{proof}

\smallskip

Finally, let us review how sequential topologies behaves under a quotient. As it was proved on \cite[Remark 5.12]{FHSIso2}, given a sequential space $(X,L)$ and an equivalence relation $\sim$ defined on it, the quotient topological space $X/\sim$ (with the induced topology) is again a sequential space. In fact, it is possible to give explicitly a limit operator $L_{Q}$ whose associated topology coincides with the quotient topology in $X/\sim$ in the following way:

\begin{equation}\label{deflimitQ}[x]\in L_{Q} (\{[x_n]\}_{n})\iff \exists\;\;
x'\in i^{-1}([x]),\; x'_n\in i^{-1}([x_n])\;\,\forall n\in\N\;
:\,\; x'\in L(\{x'_n\}_{n}).\end{equation} where $i:X\rightarrow X/\sim$ is the natural quotient projection and $[x],[x_n]\in
X/\sim$.

\subsection{C-boundary construction}
\label{sec:2.1}

The causal completion was firstly introduced by Geroch, Kronheimer and Penrose in their seminal work \cite{GKP}. The main idea for such a construction is to attach for any future-past inextensible
timelike curve an ideal point characterized by the past-future of the curve. The original construction presents several problems mainly related with the topology considered. However, the notion of causal boundary and completion have been widely developed \cite{BS,H1,H2,Ra,Sz,Sz2} (see also the reviews in \cite{GpScqg05,S}), reaching a definition for the causal completion (named c-completion) fully satisfactory on \cite{FHSFinalDef}.
%
 
Let us review some classical concepts of causal theory, referring the reader to \cite{Wald} for further details and classical notation. Let $(V,\mathfrak{g})$ be a connected, time-oriented Lorentz manifold. Denote by $\ll$ the chronological relation (resp. $\leq$ the causal relation), that is, $p\ll q$ ($p\leq q$) iff there exists a future-directed timelike (causal) curve from $p$ to $q$. In what follows, the spacetime $V$ will be considered strongly causal, and so, the intersections between the chronological future and past of points generate the topology in $V$. In particular, strong causality ensures also that $V$ is {\em distinguishing}, hence two different points $p,q\in V$ have different future $I^+(p)\neq I^+(q)$ and past $I^-(p)\neq I^-(q)$.

A non-empty
subset $P\subset V$ is called a {\em past set} if it coincides
with its past, i.e., $P=I^{-}(P):=\{p\in V: p\ll q\;\hbox{for
	some}\; q\in P\}$. Let $S\subset V$ and define the {\em common past} of $S$ as
$\downarrow S:=I^{-}(\{p\in V:\;\; p\ll q\;\;\forall
q\in S\})$. Observe that, from definition, the past and common past sets are open.
A past set that cannot be written as the union of two proper past
sets
is called {\em indecomposable past} set, {\em IP} for short. An indecomposable past set $P$ belongs to one of the following two categories: $P$ can be expressed as the past of a point of the spacetime, i.e., $P=I^-(p)$ for some $p\in V$, and so, $P$ is called {\em proper indecomposable past
	set}, {\em PIP}; or $P=I^{-}(\{x_n\}_{n})$ for some
inextendible future-directed chronological sequence $\{x_n\}_{n}$\footnote{Here, by a future-directed chronological sequence $\{x_n\}_{n}$ we mean that $x_n\ll x_{n+1}$ for all $n$, see \cite{Flores} for this approach of the causal boundary.}, and then $P$ is
called {\em terminal indecomposable past set}, {\em TIP}. The dual
notions of {\em future set}, {\em common future}, {\em IF}, {\em PIF} and {\em
	TIF}, are defined just by interchanging the roles of
past and future in previous definitions.

The future causal completion is defined as the set of all indecomposable past sets IPs. As the manifold $V$ is distinguishing, the original manifold points $p\in V$ are naturally identified with their past $p\equiv I^-(p)$, and so, $V$ is identified with the set of PIPs. Therefore, the future causal boundary $\hat{\partial} V$ is defined as the set of all TIPs in $V$, obtaining the following identifications:
 
\[
V\equiv \hbox{PIPs},\qquad \hat{\partial}V\equiv
\hbox{TIPs},\qquad\hat{V}\equiv \hbox{IPs}.
\]
The future causal completion will be endowed with the {\em future chronological topology} $\hat{\tau}_{chr}$, a sequential topology defined by the following limit operator: for $\sigma=\{P_n\}_{n} \subset \hat{V}$, 

\begin{equation}\label{def:futurelimit}
P\in \hat{L}_{chr}(\{P_n\}_{n}) \iff P\subset 
		{\rm LI}(\{P_n\}_{n}) \hbox{ and it is maximal in } 
		{\rm LS}(\{P_n\}_{n}).
\end{equation}
Here. ${\rm LS}$ and ${\rm LI}$ denotes superior and inferior limits for sets; and {\em maximal} means that no other $P'\in\hat{V}$ (resp.
$F'\in \check{V}$) satisfies the stated property and includes
strictly $P$.

An analogous definition follows for the {\em past causal completion} by interchanging the roles of future and past sets. Hence, 

\[
V\equiv \hbox{PIFs},\qquad \check{\partial}V\equiv
\hbox{TIFs},\qquad\check{V}\equiv \hbox{IFs},
\]
and $\check{V}$ is endowed with the {\em past chronological topology} $\check{\tau}_{chr}$ defined by a limit operator $\check{L}_{chr}$.

\smallskip 

For the (total) c-boundary, we need to recall that some IPs and IFs represent naturally the same point of the completion. This is quite evident for PIPs and PIFs, where future and past sets can be identified if they are future and past of the same point respectively. However, previous identification is insufficient, as other indecomposable sets have to be identified. For this, let us define the so-called S-relation (introduced on \cite{Sz}).
Denote by $\hat{V}_{\emptyset}=\hat{V}\cup \{\emptyset\}$ and by
$\check{V}_{\emptyset}=\check{V}\cup \{\emptyset\}$. The
S-relation $\sim_S$ is defined in $\hat{V}_{\emptyset}\times
\check{V}_{\emptyset}$ as follows. A pair $(P,F)\in
\hat{V}\times \check{V}$ is $S$-related if
\begin{equation} \label{eSz}  P\sim_S F \Longleftrightarrow \left\{
\begin{array}{l}
P \quad \hbox{is included and is a maximal IP into} \;
\downarrow F
\\
F \quad \hbox{is included and is a maximal IF into} \; \uparrow
P.
\end{array} \right.
\end{equation}
As proved by Szabados
\cite{Sz}, the past and future of a point $p\in V$ are $S$-related, $I^-(p) \sim_S I^+(p)$, and these are
the unique S-relations (according to our definition (\ref{eSz}))
involving proper indecomposable sets. We also define that $P\sim_S\emptyset$ (resp. $\emptyset\sim_{S} F$) if $P$ (resp. $F$) is a non-empty, necessarily terminal
indecomposable past (resp. future) set that  is not S-related by
(\ref{eSz}) to any other indecomposable set (note that
$\emptyset$ is never S-related to itself).

Now, we can introduce the notion of c-completion. At the point set level, and following the idea of Marolf and Ross \cite{MR}, the c-completion is formed by $S$-related pairs of indecomposable sets 

		\begin{equation}
		\label{def:cau1}
		\overline{V}:=\{(P,F)\in \hat{V}_{\emptyset}\times
		\check{V}_{\emptyset}: P\sim_S F\}.
		\end{equation}
		
Every point $p\in V$ of the manifold is naturally identified with its
		corresponding pair $(I^-(p),I^+(p))$, so $V$ can be (and will be) considered a
		subset of $\overline{V}$. The {\em c-boundary} is then
		defined as $\partial V=\overline{V}\backslash V$.
		The chronological relation on $V$ is also extended to the c-completion in the following way (by abuse of notation, we denote the chronological relation on $\overline{V}$ with the same symbol): given two points $(P,F),(P',F')\in\overline{V}$
		
		\begin{equation}\label{eq:defcronorel}
		(P,F)\ll (P',F')\,\, \iff\,\, F\cap 
		P'\neq\emptyset.
		\end{equation}
		
		Finally, $\overline{V}$ is endowed with the {\em chronological topology} $\tau_{chr}$, a sequential topology associated to the following limit operator (known as the chronological limit): for a sequence $\sigma=\{(P_n,F_n)\}_{n} \subset \overline{V}$, define 

		\begin{equation}\label{eq:deflimcrono}
		\lcrono(\sigma):=\left\{(P,F)\in \overline{V}:\begin{array}{l}
		P\in \hat{L}_{chr}(\{P_n\}_{n}) \hbox{ if $P\neq \emptyset$}\\
		F\in \check{L}_{chr}(\{F_n\}_{n}) \hbox{ if $F\neq \emptyset$}
		\end{array}
		\right\}.
		\end{equation}
		It is important to recall, as it will be used later, that due the definition of the $S$-relation between terminal sets, the definition of the chronological limit is simplified when both terminal sets on the limit are non empty (see \cite[Lemma 3.15]{FHSFinalDef}). Concretely, if $P\neq\emptyset\neq F$:
		
		\begin{equation}\label{eq:limitsimple}
		(P,F)\in L_{chr}(\{(P_n,F_n)\}_{n})\iff P\subset {\rm LI}(\{P_n\}_{n}) \hbox{ and } F\subset{\rm LI}(\{F_n\}_{n})
		\end{equation}

		The following result will summarize the main properties of the c-completion endowed with the chronological relation and topology (see \cite[Theorem 3.27]{FHSFinalDef} and its proof).
		
		\begin{thm}\label{teo:propcomple}
		Let $(V,\mathfrak{g})$ be a strongly causal Lorentzian manifold and $\overline{V}$ its causal completion endowed with the chronological structure induced by \eqref{eq:defcronorel} and the topology induced from the chronological limit \eqref{eq:deflimcrono}. Then:
		
		\begin{itemize}
			\item[(i)] The inclusion $V\hookrightarrow \overline{V}$ is continuous. Moreover, the restriction of the chronological limit on $V$ is a first order limit operator.
			
			\item[(ii)] 
			 Let $\{x_n\}_{n}\subset V$ a future (resp. past) chronological sequence . Then,
			
			\[
			L_{chr}(\{x_n\}_{n})=\{(P,F)\in \overline{V}:\,P=I^-(\{x_n\}_{n}) \,\hbox{(resp. $F=I^+(\{x_n\}_{n})$)}\}
			\]
			
			\item[(iii)]  The c-completion is complete: For any past terminal set $P$ (resp. future terminal set $F$) there exists $F$ (resp. $P$) such that $(P,F)\in \overline{V}$. In particular, any inextensible timelike curve $\gamma$ on $V$ (resp. any inextensible chronological sequence $\{x_n\}_{n}$ on $V$) has an endpoint in $\overline{V}$.
			
			\item[(iv)] The sets $I^\pm ((P,F))$ are open for all $(P,F) \in \overline{V}$.
			
			\item[(v)] $\overline{V}$ is a $T_1$ topological space.
			
		\end{itemize} 
		\end{thm}
%
%

\subsection{Spacetime covering projections: The causal ladder and main properties}
\label{sec:spaceproy}
Let us consider that we have an action on $V$ given by a group $G$ of isometric maps\footnote{The results of this paper can be obtained considering a group $G$ of conformal maps, but we will consider isometric actions for simplicity.}.
 
 \[\begin{array}{ccc} 
 G\times V &\rightarrow& V\\ (g,p) & \rightarrow &g\,p. 
 \end{array} 
 \]

We will always assume that the action preserve time-orientation, and acts freely and {\em properly discontinuously}, where the latter means: (a) for each $p\in V$, there exists a neighborhood $U$ such that $g\,U\cap U=\emptyset$ for all $g\in G\setminus \{e\}$ and; (b)  for $p_1,p_2\in V$ there are neighbourhoods $U_1$ and $U_2$ such that $g\,U_1\cap U_2=\emptyset$ for all $g\in G$.
%
%
%

Previous conditions over the action let us ensure that the quotient space $M=V/G$ is also a Lorentzian manifold with the induced metric (which will be denoted by an abuse of notation as $\mathfrak{g}$). The canonical projection to the quotient space, denoted by $\pi:V\rightarrow M$, will be called a {\em spacetime covering projection}. The following result let us understand clearly the relation between the chronological relation on $M$ and $V$ (the same result follows for causal relations, see \cite[Proposition 1.1]{H}). 

\begin{prop}\label{prop:caracrel}
	Let us consider $\pi:V\rightarrow M$ a spacetime covering projection. Then:
	
	\begin{itemize}
		\item If $p,q\in V$ satisfy that $p\ll q$, then $\pi(p)\ll \pi(q)$.
		\item If $x,y\in M$ satisfy that $x\ll y$, then for any $p,q\in V$ with $\pi(p)=x$ and $\pi(q)=y$, there exists an element $g\in G$ such that $p\ll g\,q$. 
	\end{itemize}
\end{prop}

As it is clear, previous result is key to understanding the relation between the causal structures of both, $V$ and $M$. From a global viewpoint, it is possible to characterize all the stages of the well known causal ladder on $M$ (see \cite{MS}) in terms of the global causal structure of $V$. We will summarize in the following result some of such characterizations, which proofs can be found on \cite[Propositions 1.2, 1.3 and 1.4]{H}.

\begin{thm}\label{thm:causalladder}
	Let $\pi:V\rightarrow M$ be a spacetime covering projection with group $G$. Then:
	
	\begin{itemize}
		\item[(CL1)] $M$ is non-totally vicious if, and only if, there exists $p,q\in V$ with $\pi(p)=\pi(q)$ and $p\not\ll q$.
		\item[(CL2)] $M$ is chronological if, and only if, for all $p,q\in V$ with $\pi(p)=\pi(q)$, $p\not\ll q$.
		
		\item[(CL3)] $M$ is causal if, and only if, for all $p,q\in V$ with $\pi(p)=\pi(q)$,  $p\not\leq q$.
		
		
		\item[(CL4)] $M$ is strongly causal if, and only if, for all $p\in V$ there is a fundamental neighbourhood system $\{U_n\}$ for $p$ such that for each $n$, no causal curve can have one endpoint in $U_n$ and another endpoint in a component $U'_n$ of $\pi^{-1}(\pi(U_n))$ unless $U'_n=U_n$ and the curve remains wholly within $U_n$. 
		
	
	
		
		\item[(CL5)] $M$ is globally hyperbolic if, and only if,
		\begin{itemize}
			\item[(CL5-1)] $V$ is globally hyperbolic,
			\item[(CL5-2)] every point $p\in V$ has a fundamental neighbourhood system as in (CL4) and
			\item[(CL5-3)] for any $p\in V$, for all $p\ll q$, $J^+(p)\cap \pi^{-1}(\pi(q))$ is finite. 
		\end{itemize}
	
	\end{itemize}
\end{thm}

Let us remark that in all previous cases, a global causal condition on $M$ (i.e., the assumption of a stage in the causal ladder) implies a stronger global condition on $V$. However at this point, it is not clear for us at what extent the same property follows for the rest of the causal ladder (particularly with {\em causally continuous} and {\em causally simple}), being necessary a detailed study on such cases. However, that study is out the scope of this paper.

\section{Partial Boundaries under the action of the Group}
\label{sec:Partial}
In this section, we will study the behaviour of the future causal completion under the action of an isometry group $G$, being the past case completely analogous. Let us begin with a point in the future completion of $V$, that is, an indecomposable set $\lev{P}{}=I^-(\{p_n\}_{n})$, where $\{p_n\}_{n}$ is a future-directed chronological sequence. As the group $G$ acts by isometries in $V$, the sequence $\{x_n\}_{n}$ with $\pi(p_n)=x_n$ is also future-directed and chronological (Prop. \ref{prop:caracrel}), hence, it defines the indecomposable set $P=I^-(\{x_n\}_{n})$ in $M$. Therefore, the projection $\pi$ extends naturally to the corresponding partial completions on the following way: 

\begin{equation}\label{eq:aux4} 
\begin{array}{lccc}
\hat{\pi}:&\hat{V}&\rightarrow& \hat{M}\\
 & \lev{P}{}=I^-(\{p_n\}_{n})&\rightarrow & P=I^-(\{x_n\}_{n}).
\end{array}
\end{equation}
We will say that an indecomposable set $\lev{P}{}\in \hat{V}$ is a lift of $P$ if $\hat{\pi}(\lev{P}{})=P$. 

Previous map is always surjective, as any future-directed chronological sequence $\{x_n\}_{n}$ in $M$ can be lifted to a future-directed chronological sequence $\{p_n\}_{n}$ in $V$ (by Prop. \ref{prop:caracrel}).
However, the map is not injective in general, as previous lift is not unique. For instance, if $\{p_n\}_{n}$ is a lift of $\{x_n\}_{n}$, $\{g\,p_n\}_{n}$ is also a lift of the same sequence. 
Even more, the pre-image of a terminal set $P$ can be easily characterized. Let us denote by $\lev{P}{}=I^-(\{p_n\}_{n})$, where $\{p_n\}_{n}$ denotes one fixed lift of $\{x_n\}_{n}$. It follows that
%

\[
\hat{\pi}^{-1}(P)=\cup_{g\in G}g\,\lev{P}{},
\]
i.e., the pre-image of $P$ is the union of what we are going to call the {\em  $G$-orbit of $\lev{P}{}$ in  $\hat{V}$}, which is the set $\{g\,\lev{P}{}\}_{g\in G}$. The left inclusion is straightforward, as $\pi(g\,\lev{P}{})=P$ for all $g\in G$. For the right one, take a point $x\in P$ and let $p\in V$ be a point such that $\pi(p)=x$. As $x\in P$, there exists $n$ big enough such that $x\ll x_n$. Hence, Prop. \ref{prop:caracrel} ensures that $p\in g\,p_n\subset g\,\lev{P}{}$.

\begin{conv}\label{conv} 
	{\em From this point, there are some useful conventions that we will use along the paper. For instance, the points on $M$ will be denoted by $x,y,z$, while the points on $V$ will be denoted by $p,q,r$. Moreover, unless stated otherwise, we will always assume that $\pi(p)=x, \pi(q)=y$ and $\pi(r)=z$.
	
	For any chronological sequence $\{x_n\}_{n}$ in $M$ (resp., an indecomposable set $P$), we will consider a fixed lift on $V$ denoted by $\{p_n\}_{n}$ (resp. $\lev{P}{}$). As an abuse of notation, we will use the same symbol $I^\pm$ for future/past of sets when there is no confusion if we are in $M$ or $V$.

Finally, and in order to compute both, partial and c-boundary, we will assume from this point that $M$ is strongly causal and so that $V$ satisfies the condition described on Theorem \ref{thm:causalladder} (CL4).
}
\end{conv} 

The projection $\hat{\pi}$ let us define an equivalence relation on $\hat{V}$: two indecomposable sets $\lev{P}{1},\lev{P}{2}\in {\hat{V}}$ are $\hat{G}$-related,  $\lev{P}{1}\sim_{\hat{G}} \lev{P}{2}$, if and only if both terminal sets projects onto the same $P\in \hat{M}$, i.e., $\hat{\pi}(\lev{P}{1}) = \hat{\pi}(\lev{P}{2})$. Of course, previous relation lead us to a bijection between the quotient space $\hat{V}/\hat{G}(\equiv \hat{V}/\sim_{\hat{G}})$ and $\hat{M}$. However, the following two observations are in order: On the one hand, one could expect naively that for any two terminal sets with $\lev{P}{1}\sim_{\hat{G}}\lev{P}{2}$, there exists $g\in G$ such that $\lev{P}{1}=g\,\lev{P}{2}$. This is not true in general, as it can be observed on Example \ref{ex:tame}. In fact, such a property motivates the following definition. 

\begin{defi}\label{rem:tameextension}
	We will say that a spacetime covering projection $\pi:V\rightarrow M$ is {\em future tame}  if given two terminal sets $\lev{P}{1},\lev{P}{2}$ with $\lev{P}{1}\sim_{\hat{G}}\lev{P}{2}$ there exists $g\in G$ such that $\lev{P}{1}=g\,\lev{P}{2}$.  
\end{defi}

On the other hand, the induced map is not well behaved at the topological level. In fact, Harris shows in the last example of  \cite{H} that $\hat{\pi}$ is not, in general, continuous (see also Example \ref{ex:exemHarris} for details).


\smallskip

The rest of this section is devoted to make a deep comparison between the topologies of $\hat{M}$ and $\hat{V}/\hat{G}$, where the latter has the induced quotient topology. 
Let us first fix some notation. As we have mention on Section \ref{sec:2.1}, $\hat{M}$ and $\hat{V}$ will be endowed with the future chronological topology, which is defined by a limit operator \eqref{def:futurelimit}. In order to differentiate both limits, we will denote by $\hat{L}_{M}$ the future chronological limit on $\hat{M}$ and, accordingly, $\hat{L}_{V}$ the limit on $\hat{V}$. 
The quotient topology on $\hat{V}/\hat{G}$ is also a sequential topology (see Section \ref{sec:prellimits}) and it is defined from a limit operator \eqref{deflimitQ} which will be denoted here by $\hat{L}_{\hat{G}}$. Finally, recall that the map $\hat{\pi}$ induces a bijective map $\hat{\j}$ between $\hat{V}/\hat{G}$ and $\hat{M}$ which makes the following diagram commutative:

\begin{center} 
\begin{tikzpicture}
\matrix (m) [matrix of math nodes,row sep=3em,column sep=4em,minimum width=2em]
{
	\hat{V} & \\
	\hat{V}/\hat{G} & \hat{M} \\};
\path[-stealth]
(m-1-1) edge node [left] {$\hat{\i}$} (m-2-1)
(m-2-1.east|-m-2-2) edge node [below] {$\hat{\j}$} (m-2-2)
(m-1-1) edge node [above] {$\hat{\pi}$} (m-2-2);
\end{tikzpicture}
\end{center}
where $\hat{\i}:\hat{V}\rightarrow \hat{V}/\hat{G}$ is the usual projection.

\smallskip 

Previous $\hat{\jmath}$ map is always open. In order to prove this, we require the following technical lemma.

\begin{lemma}
	Consider a sequence $\sigma=\{P_n\}_{n}\subset \hat{M}$ and a point $P\in \hat{M}$ such that $P\subset {\rm LI}(\{P_n\}_{n})$. For $\lev{P}{}$ a fixed lift of $P$, there exist lifts $\lev{P}{n}$ of $P_n$ such that $\lev{P}{}\subset {\rm LI}(\{\lev{P}{n}\}_{n})$. 
\end{lemma}

\begin{proof}
	Let us begin by taking $\{\lev{P}{n}\}_{n}$ some fixed lifts of $\{P_n\}_{n}$. Denote also by $\{p_n\}_{n}$ and $\{x_n\}_{n}$ future chronological chains 
	defining $\lev{P}{}$ and $P$ resp. and satisfying that $\pi(p_n)=x_n$ (as stated in Convention \ref{conv}). As $P\in {\rm LI}(\{P_n\}_{n})$, for any element $x_n$ there exists $m_n\in\mathbb{N}$ (that we can consider strictly increasing on $n$) such that, for all $m\geq m_n$, $x_n\in P_m$. In particular, and due to Prop. \ref{prop:caracrel}, we can ensure the existence of $g\in G$ in such a way that $p_n\in g\, \lev{P}{m}$. Then, for $m\geq m_n$, let us denote by $G(n,m)\subset G$ the non-empty subset defined in the following way:
	
	\begin{equation}\label{eq:Gnm}
	G(n,m):=\{g\in G: p_n\in g\,\lev{P}{m}\}
	\end{equation}
	
	Let us make a straightforward (but necessary) observation about previous sets. As $p_n\ll p_{n+1}$, for $m\geq m_{n+1}(\geq m_n+1)$,
	
	\begin{equation}\label{eq:orderGnm}
	G(n+1,m)\subset G(n,m).
	\end{equation}   
	
	Now, for each $m_{n}\leq m < m_{n+1}$, let us consider a group element  $g_m\in G(n,m)$ and consider the sequence $\{g_{m}\,\lev{P}{m}\}_{m}$ (for $m<m_1$, just consider $g_m=e$, the identity). Now, let us show that previous sequence is the desired, that is, $\lev{P}{}\subset {\rm LI}(\{g_m\,\lev{P}{m}\}_{m})$. In fact, for any $n\in\mathbb{N}$, consider $m\geq m_{n}$ and denote by $k\in \mathbb{N} \cup \{0\}$ the natural ensuring that $m_{n+k+1} > m\geq m_{n+k}$. Then, from the choice of $\{g_m\}_{m}$ and (\ref{eq:orderGnm}), we have that:
	
	$$g_{m}\in G(n+k,m)\subset G(n+k-1,m)\subset\dots \subset G(n,m).$$
	In conclusion, from (\ref{eq:Gnm}) we deduce that $p_n\in g_{m}\,\lev{P}{m}$ for all $m\geq m_n$, and the result follows.  
\end{proof}

\begin{prop}\label{thm:openhatj}
Let $\pi:V\rightarrow M$ be a spacetime covering projection and $\hat{\pi}:\hat{V}\rightarrow \hat{M}$ the extended map on the corresponding partial completions. The induced map $\hat{\j}:\hat{V}/\hat{G}\rightarrow \hat{M}$ is open.
\end{prop}
\begin{proof}
Let us prove that the map $\hat{\jmath}^{-1}$ is continuous by using Prop. \ref{prop:limcont}. For this, consider a sequence $\sigma=\{P_n\}_{n} \subset \hat{M}$ and a point $P\in \hat{L}_{M}(\sigma)$, and let us show that $\hat{\jmath}^{-1}(P)\in \hat{L}_{\hat{G}}(\hat{\jmath}^{-1}(\kappa))$ for some subsequence $\kappa\subset \sigma$. Recall that, from the definitions of $\hat{L}_{\hat{G}}$ and $\hat{\jmath}^{-1}$, this is the same that show the existence of lifts $\lev{P}{n}$ and $\lev{P}{}$ of $P_n$ and $P$ resp. such that $\lev{P}{}\in \hat{L}_V(\{\lev{P}{n_k}\}_{k})$ for some subsequence $\{\lev{P}{n_k}\}_{k} \subset \{\lev{P}{n}\}_{n}$.   

First observe that, by using previous lemma, we can find lifts $\lev{P}{n}$ and $\lev{P}{}$ of $P_n$ and $P$ resp. such that $\lev{P}{}\subset {\rm LI}(\{\lev{P}{n}\}_{n})$. If $\lev{P}{}$ is maximal in ${\rm LS}(\{\lev{P}{n}\}_{n})$, then we have that $\lev{P}{}\in \hat{L}_{V}(\{\lev{P}{n}\}_{n})$, and we are done. 

Otherwise, take $\lev{P'}{}$ a maximal set in ${\rm LS}(\{\lev{P}{n}\}_{n})$ containing $\lev{P}{}$. By definition of the superior limit, we can consider a subsequence $\{\lev{P}{n_k}\}_{k}$ in such a way that $\lev{P'}{}$ is both, contained in ${\rm LI}(\{\lev{P}{n_k}\}_{k})$ and maximal in ${\rm LS}(\{\lev{P}{n_k}\}_{k})$, i.e., $\lev{P'}{}\in \hat{L}_V(\{\lev{P}{n_k}\}_{k})$. Now observe that the sets $P'=\hat{\pi}(\lev{P'}{})$ and $P_{n_k}=\hat{\pi}(\lev{P}{n_k})$ satisfy the following chain ($\hat{\pi}$ preserves contentions)

\[
P\subset P'\subset {\rm LI}(\{P_{n_k}\}_{k}).
\]
But as $P\in \hat{L}_M(\{P_{n_k}\}_{k})$, it follows that $P=P'$ (recall the maximal character on \eqref{def:futurelimit}) and so that $\lev{P'}{}$ is also a lift of $P$.

In both cases, and up to a subsequence, we show the existence of lifts $\{\lev{P}{n}\}_{n}$ and $\lev{P}{}$ with $\lev{P}{}\in \hat{L}_V(\{\lev{P}{n}\}_{n})$, and then the continuity of $\hat{\jmath}^{-1}$ follows from Prop. \ref{prop:limcont}.

\end{proof}
\begin{rem}\label{rem:aux1}
	Previous proof shows in particular that for all $P\in \hat{L}_M(\{P_n\}_{n})$, there exist lifts $\lev{P}{}$ and $\lev{P}{n}$ of $P$ and $P_n$ resp. with $\lev{P}{}\in \hat{L}_{V}(\{\lev{P}{n_k}\}_{k})$ for some subsequence $\{\lev{P}{n_k}\}_{k}$ of $\{\lev{P}{n}\}_{n}$.
\end{rem}

\smallskip

As we have already pointed out, the map $\hat{\j}$ is not continuous in general. If we look into the details of Example \ref{ex:exemHarris}, we see that the non-continuity is related with the following situation: There exists a (non-necessarily chronological) sequence $\{P_n\}_{n} \subset \hat{M}$ admitting two different lifts such that (i) both lifted sequences are convergent and (ii) the projection of one limit point contains strictly the other. As we will see, such a situation represent, essentially, the only cases where continuity of $\hat{\pi}$ can fail, so it is convenient to give a proper name for it:

\begin{defi}\label{def:divlif}
	Let $\pi:V\rightarrow M$ a spacetime covering projection and $\hat{V},\hat{M}$ the corresponding future causal completions of $V$ and $M$. We will say that a sequence $\sigma=\{P_n\}_{n} \subset \hat{M}$ has {\em future divergent lifts} if there exist two lifts $\{\lev{P}{n}\}_{n},\{\lev{P'}{n}\}_{n} \subset \hat{V}$  of $\sigma$  and two points $\lev{P}{},\lev{P'}{}\in \hat{V}$ such that:
	\begin{itemize}
		\item[(i)] $\lev{P}{}\in \hat{L}_V(\{\lev{P}{n}\}_{n})$ and $\lev{P'}{}\in \hat{L}_V(\{\lev{P'}{n}\}_{n})$.
		
		\item[(ii)] $\hat{\pi}(\lev{P}{})\subsetneq \hat{\pi}(\lev{P'}{})$.
		
	\end{itemize}
	
	If there exists no such a sequence on $\hat{M}$, we will just say that $M$ does not admit future sequences with divergent lifts.
\end{defi}

As a side remark, observe that the concept of divergent lifts is quite related with the topological structure of the $G$-orbits in $\hat{V}$. In fact, we can prove the following result:

\begin{prop}\label{prop:new}
  If the future $G$-orbits of $\hat{V}$ are closed, then $M$ does not admit {\em constant} sequences with divergent lifts. The converse is also true if we assume additionally that $\pi$ is future tame.
\end{prop}

\begin{proof}
  For the first assertion let $P\in \hat{M}$ and $\lev{P}{}\in \hat{V}$ with $\hat{\pi}(\lev{P}{})=P$. Observe that, if we have a sequence $\{g_n\}_n\subset G$ and $\lev{P'}{}$ such that $\lev{P'}{}\in \hat{L}_{V}(\{g_n\,\lev{P}{}\}_n)$, the closedness of the $G$-orbit ensures that $\lev{P'}{}=g_0\,\lev{P}{}$ for some $g_0\in G$. Therefore,  $\{P\}$ admits no divergent lifts as condition (ii) in Def. \ref{def:divlif} cannot be fullfilled.

  For the second assertion, assume that $M$ admits no constant sequence with divergent lifts and that $\pi$ is future tame; and let us prove that the $G$-orbits in $\hat{V}$ are closed. Let $P,\lev{P}{},\lev{P'}{}$ and $\{g_n\}_n$ as in previous implication. As $M$ admits no constant sequence with divergent lifts, then necessarily it follows that $\hat{\pi}(\lev{P'}{})=P$. Moreover, as $\pi$ is future tame, then there exists $g_0\in G$ such that $\lev{P'}{}=g_0\, \lev{P}{}$, and so, $\lev{P'}{}$ belongs to the $G$-orbit $\{g\, \lev{P}{}\}_{g\in G}$ and the $G$-orbit is closed. 
 
\end{proof}
The optimality of previous result follows from Example \ref{ex:tame} where it is shown a case where $M$ admits no constant sequence with divergent lifts but the $G$-orbits are not closed.

\smallskip

Our main result concerning the continuity of the partial completions is the following:

\begin{prop}\label{thm:hatjcont}
Let $\{\lev{P}{n}\}_{n}$ be a sequence whose projection $\{P_n\}_{n}$ does not admit divergent lifts. Then,
\[
\lev{P}{}\in \hat{L}_V(\{\lev{P}{n}\}_{n}) \Rightarrow P\in \hat{L}_M(\{P_n\}_{n}).
\]
In particular, 
if $M$ does not admit sequences with future divergent lifts, 
 the map $\hat{\pi}$ is continuous. If, additionally, the future chronological limit $\hat{L}_{M}$ on $M$ is of first order UTS, then the continuity of $\hat{\pi}$ ensures the non-existence of sequences with divergent lifts.
\end{prop} 

\begin{proof} 

Let $\lev{\sigma}{}=\{\lev{P}{n}\}_{n}$ be a sequence as in the first statement of the proposition, and consider $\lev{P}{}\in \hat{L}_V(\lev{\sigma}{})$. By recalling that $\hat{\pi}$ preserves contentions, we deduce that $P\subset {\rm LI}(\{P_n\}_{n})$. If $P$ is maximal among the IPs in ${\rm LS}(\{P_n\}_{n})$, then $P\in \hat{L}_M(\{P_n\}_{n})$ and we are done. 

So, let us assume by contradiction that $P$ is not maximal on the ${\rm LS}(\{P_n\}_{n})$. Consider $P'$ a maximal set on ${\rm LS}(\{P_n\}_{n})$  containing strictly $P$. From the definition of the superior limit, and up to a subsequence, we can assume that $P' \subset {\rm LI}(\{P_n\}_{n})$, and so, that $P'\in \hat{L}_{M}(\{P_n\}_{n})$. Now, recalling Remark \ref{rem:aux1}, we ensure that $P_n$ and $P'$ admit lifts $\lev{P'}{n}$ and $\lev{P'}{}$ such that $\lev{P'}{}\in \hat{L}_{V}(\{\lev{P'}{n_k}\}_{k})$. Summarizing, the sequence $\{P_{n_k}\}_{k}$ admits two lifts $\{\lev{P}{n_k}\}_{k}$ and $\{\lev{P'}{n_k}\}_{k}$ converging to $\lev{P}{}$ and $\lev{P'}{}$ resp., where $P=\hat{\pi}(\lev{P}{})\subsetneq \hat{\pi}(\lev{P'}{})=P'$. That is to say, $\{P_{n_k}\}_{k}$ admits future divergent lifts, a contradiction.  
In conclusion, $P\in \hat{L}_{M}(\{P_n\}_{n})$. Moreover, if $M$ does not admit sequences with divergent lifts, $\hat{\pi}$ is continuous (recall Prop. \ref{prop:limcont}).
\smallskip 

For the final assertion, assume that $\hat{L}_{M}$ is of first order UTS and that there exists a sequence $\sigma=\{P_n\}_{n}\subset M$ with divergent lifts. Let $\{\lev{P}{n}\}_{n}, \{\lev{P'}{n}\}_{n}$ be two sequences in $\hat{V}$ and $\lev{P}{},\lev{P'}{}$ two terminal sets  as in Def. \ref{def:divlif}. Assume by contradiction that $\hat{\pi}$ is continuous. In particular, we have that $\{P_n\}_{n}$ (the projection by $\hat{\pi}$ of both sequences $\{\lev{P}{n}\}_{n}$ and $\{\lev{P'}{n}\}_{n}$) converges to $P$ and $P'$. As $\hat{L}_{M}$ is of first order UTS, we can assume that (up to a subsequence) $\{P_n\}_{n}$ is of first order, and so, that $P,P'\in \hat{L}_{M}(\{P_n\})_{n}$. But this is a contradiction with the definition of $\hat{L}_{M}$ \eqref{def:futurelimit} (concretely the maximal character of the limit points) and the fact that $P\subsetneq P'$ (Def. \ref{def:divlif} (ii)). Therefore, the map $\hat{\pi}$ cannot be continuous. 
\end{proof} 

There are several ways to prove the non-existence of sequences with divergent lifts. For instance, we can impose conditions on the causality of the boundary (re-obtaining \cite[Theorem 3.4]{H})



\begin{cor}\label{cor:Harrisresult}
If $M$ has only spatial future boundary points, then $\hat{\pi}$ is continuous, and so, $\hat{\jmath}$ is an homeomorphism between $\hat{M}$ and $\hat{V}/\hat{G}$.
\end{cor}

\begin{proof}
	Assume by contradiction that $\hat{\pi}$ is not continuous and so, from previous result, that there exists a sequence $\sigma \in \hat{M}$ admitting divergent lifts. Let $\overline{\sigma},\overline{\sigma'}$ be two sequences in $\hat{V}$ and $\lev{P}{}, \lev{P'}{}$ be two points in $\hat{V}$ as in Def. \ref{def:divlif}. As $\hat{M}$ only contains spatial future boundary points, no IP can contain a TIP. Hence, from (ii) in Def. \ref{def:divlif}, we deduce that  $P=I^-(x)$ for some $x\in M$, and then, $\lev{P}{}=I^-(p)$ for some point $p\in V$. As $\pi:V\rightarrow M$ is continuous and the future chronological topology preserves the manifold topology (which follows from Thm. \ref{teo:propcomple}, (i)), we have that $P\in \hat{L}_M(\{P_n\}_{n})$. Finally from (i) and (ii) in Def. \ref{def:divlif} we have that $P\subsetneq P'\in {\rm LI}( \{P_n\}_{n})$, in contradiction with the maximality on \eqref{def:futurelimit}.  
\end{proof}

Another possibility is to impose conditions over the topology of the future causal completion. In this case, we have also need to impose the finiteness of the group action $G$:

\begin{cor}\label{cor:aux1}
	Consider $\pi:V\rightarrow M$ a spacetime covering with associated group $G$. Assume that $G$ is finite and that $\hat{V}$ is Hausdorff. Then, $\hat{\pi}$ is continuous, and so, $\hat{\jmath}$ is an homeomorphism.
\end{cor}

\begin{proof}
	As we will see in the forthcoming sections, if $G$ is finite then $\pi$ is future tame (see Lemma \ref{lem:aux1}). Hence, let us consider two sequences $\{\lev{P}{n}\}_{n},\{\lev{P'}{n}\}_{n} \subset \hat{V}$ and two points $\lev{P}{},\lev{P'}{}\in \hat{V}$ with $\lev{P}{}\in \hat{L}_{V}(\{\lev{P}{n}\}_{n})$ and $\lev{P'}{}\in \hat{L}_{V}(\{\lev{P'}{n}\}_{n})$ and such that $\hat{\pi}(\lev{P}{n})=\hat{\pi}(\lev{P'}{n})$. Our aim is to prove that $\hat{\pi}(\lev{P}{})=\hat{\pi}(\lev{P'}{})$ as then no sequence with divergent lifts can exists.
	
	 Recalling the tameness of $\pi$, there exists a sequence $\{g_n\}_{n} \subset G$ such that $\lev{P'}{n}=g_n\,\lev{P}{n}$. Due the assumption that $G$ is finite, we can assume (up to a subsequence) that $g_n\equiv g_0$ for all $n$ and some constant $g_0\in G$. Therefore, $\lev{P}{}\in \hat{L}_{V}(\{\lev{P}{n}\}_{n})$ and $\lev{P'}{}\in \hat{L}_{V}(\{g_0\,\lev{P}{n}\}_{n})$. From the first inclusion and the fact that $G$ acts by isometries, we deduce that $g_0\,\lev{P}{}$ also belong to $\hat{L}_{V}(\{g_0\,\lev{P}{n}\}_{n})$ and recalling that $\hat{V}$ is Hausdorff (and so, for any sequence $\sigma$, $\hat{L}_V(\sigma)$ can contain at most one element, recall \eqref{eq:1}), it follows that $g_0\,\lev{P}{}=\lev{P'}{}$, as desired.
\end{proof}

%

\subsection{Proof of Theorem \ref{thm:main1}} Assertion (i) follows from Prop. \ref{thm:openhatj}, while (ii) from Prop. \ref{thm:hatjcont}. The last assertion is proved on Corollaries \ref{cor:Harrisresult} and \ref{cor:aux1}. 


\section{The C-completion under the action of the Group}\label{sec:total} 

Once we have determined the requirements to ensure the well behaviour of the partial boundaries, we are in conditions to study the (total) c-completion. As a first step, we will deal with the projection and lift of points of the corresponding c-completions, in order to define an extension $\overline{\pi}:\overline{V}\rightarrow \overline{M}$. Later, we will study the properties of such a map at both, the chronological and the topological level. 

\subsection{Point set level}\label{sec:totalpoint} Let us begin by considering $\lev{P}{}\in \hat{V}$ and $\lev{F}{}\in \check{V}$ two non empty terminal sets which are S-related; and let us denote by $\{p_n\}_{n}$ and $\{q_n\}_{n}$ the corresponding (future and past resp.) chronological sequences defining them. From the definition of the $S$-relation and the chronological limits, it follows that $\lev{P}{}\in \hat{L}_V(\{I^{-}(q_n)\}_{n})$ (see Thm \ref{teo:propcomple} (ii)). 
If the past chronological sequence $\{y_n\}_{n}$ (projection of $\{q_n\}_{n}$) does not admit future divergent lifts, then Prop. \ref{thm:hatjcont} ensures that $P\in \hat{L}_M(\{I^-(y_n)\}_{n})$. Then, taking into account that $\{y_n\}_{n}$ determines $F$, we obtain that $P\subset \downarrow F$ and it is maximal inside such a subset (see \eqref{def:futurelimit}). Analogously, assuming that the future chronological chain $\{x_n\}_{n}$ does not admit sequences with past divergent lifts, we can prove that $F\subset \uparrow P$ and it is maximal, so we have that:

%
 
\begin{prop}\label{prop:proypares}
	Let $\pi:V\rightarrow M$ be a spacetime covering projection. Assume that $M$ does not admit past chronological sequences with future divergent lifts nor future chronological sequences with past divergent lifts. If $(\lev{P}{},\lev{F}{})\in \overline{V}$ with $\lev{P}{}\neq\emptyset\neq\lev{F}{}$, then $(P,F)\in \overline{M}$, where $P=\hat{\pi}(\lev{P}{})$ and $F=\check{\pi}(\lev{F}{})$.
\end{prop}

At the point set level, previous proposition is the only case where points are well projected in general. In fact, Example \ref{ex:exe1} and \ref{ex:exe2} show cases of points in $\overline{V}$ with no natural projection in $\overline{M}$. Moreover, these examples also show that the lifts of points from $\overline{M}$ are not, in general, well behaved either. Concretely, as we can see in Example \ref{ex:exe2}, the point $(P_2,\emptyset)$ has no natural lift in $\overline{V}$. The only possible candidate is the point $(\lev{P}{2},\lev{F}{})$, but we are including causal information due to $\lev{F}{}$ and different problems, even at the topological level, appear.

However, if we characterize the conditions under which the lift of points $(P,F)\in \overline{M}$ with both components non empty is well defined, then we will be in conditions to define the projection between $\overline{V}$ and $\overline{M}$.

\begin{prop}\label{prop:levpares}
Consider a point $(P,F)\in \overline{M}$ with $P\neq\emptyset\neq F$. The point $(P,F)$ has a {\em lift} in $\overline{V}$, i.e., a pair $(\lev{P}{},\lev{F}{})\in\overline{V}$ with $P=\hat{\pi}(\lev{P}{})$ and $F=\check{\pi}(\lev{F}{})$ if and only if there exist lifts $\lev{P'}{}$ and $\lev{F'}{}$ of $P$ and $F$ resp. such that $\lev{P'}{}\subset \downarrow \lev{F'}{}$ and $\lev{F'}{}\subset \uparrow \lev{P'}{}$.
\end{prop}

\begin{proof}
The right implication is trivial, so we only need to focus on the left one, that is, consider a point $(P,F)\in \overline{M}$ and suppose that there exist lifts $\lev{P'}{}$ and $\lev{F'}{}$ such that $\lev{P'}{}\subset \downarrow \lev{F'}{}$ and $\lev{F'}{}\subset\uparrow\lev{P'}{}$. 
We can ensure then the existence of an IP $\lev{P}{}$ with $\lev{P'}{}\subset\lev{P}{}$ and maximal among the indecomposable sets contained in $\downarrow \lev{F'}{}$. Recalling that the projection is well behaved with contentions, we deduce that $P\subset \hat{\pi}(\lev{P}{})\subset \downarrow F$. However $P\sim_{S}F$, so the maximality on \eqref{eSz} implies that $P=\hat{\pi}(\lev{P}{})$.

Reasoning in the same way, we can prove that there exists $\lev{F}{}$ with $\check{\pi}(\lev{F}{})=F$ and being a maximal IP contained in $\uparrow \lev{P}{}$. In conclusion, $\lev{P}{}\sim_S\lev{F}{}$ and the pair $(\lev{P}{},\lev{F}{})$ belongs to $\overline{V}$. Moreover, from construction $\hat{\pi}(\lev{P}{})=P$ and $\check{\pi}(\lev{F}{})=F$, as desired. 
\end{proof}

\begin{rem}\label{rem:conjuntosSrelated}
	Recall that previous proof does not imply that the initial $\lev{P'}{}$ and $\lev{F'}{}$ are $S$-related, but that there exist others indecomposable sets $S$-related $\lev{P}{}$ and $\lev{F}{}$ such that: (a) $\lev{P'}{}\subset \lev{P}{}$, $\lev{F'}{}\subset \lev{F}{}$ and (b) $\hat{\pi}(\lev{P}{})=\hat{\pi}(\lev{P'}{})$ and $\check{\pi}(\lev{F}{})=\check{\pi}(\lev{F'}{})$.
	
\end{rem}

Now, we are in conditions to extend the projection to the c-completions. 
However, the definition of the projection is more technical than the partial cases. In fact, we have to begin by defining some identifications on $\overline{V}$ a priori, and a completely explicit definition for the projection will not be provided on the general case. When we suppose additionally that the projection is tame\footnote{As a general convention, when there is no mention to the future nor past of the definition, it will mean that we are considering both concepts at the same time.} (Def. \ref{rem:tameextension}), all previous technicalities disappear, being able to follow an analogous procedure to the definition of $\hat{\pi}$ and $\check{\pi}$.     	
	
\smallskip 	
	
Let us start by defining an equivalence relation on $\overline{V}$. Such a relation $\sim_{G}$ is defined as the smallest equivalence relation satisfying that $(\lev{P}{},\emptyset)\sim_{G} (\lev{P'}{},\lev{F'}{})$ (resp. $(\emptyset,\lev{F}{})\sim_{G} (\lev{P'}{},\lev{F'}{})$) if $\hat{\pi}(\lev{P}{})=\hat{\pi}(\lev{P'}{})$ (resp. $\check{\pi}(\lev{F}{})=\check{\pi}(\lev{F'}{})$). Denote by $\overline{V}/G$ the associated quotient space and $\overline{\i}:\overline{V}\rightarrow \overline{V}/G$ the natural projection to the quotient space. We will endow $\overline{V}/G$ with the quotient topology. Then, 

\begin{defi}\label{def:proycompleto}
Consider $\pi:V\rightarrow M$ be a spacetime covering projection. Take $(\lev{P}{},\lev{F}{})\in \overline{V}$ and define 

\begin{equation}\label{def:proy}
\overline{\pi}(\lev{P}{},\lev{F}{}):=\left\{\begin{array}{ll}
(\hat{\pi}(\lev{P'}{}),\check{\pi}(\lev{F'}{})) & \hbox{if there exists $(\lev{P'}{},\lev{F'}{})\in \overline{V}$ with}\\ &  \hbox{$\lev{P'}{}\neq\emptyset\neq \lev{F'}{}$ and $(\lev{P}{},\lev{F}{})\sim_{G}(\lev{P'}{},\lev{F'}{})$.}\\ & \\
(\hat{\pi}(\lev{P}{}),\check{\pi}(\lev{F}{})) & \hbox{Otherwise}
\end{array}   \right.
\end{equation}
%
where we are defining $\hat{\pi}(\emptyset)=\emptyset$ and $\check{\pi}(\emptyset)=\emptyset$.	
\end{defi} 
 
 This definition means that, among the elements of an equivalence class of $\overline{V}/G$, we will prioritize the projection of elements with both components non empty. It is straightforward that, with this definition, 

\[
(\lev{P}{},\lev{F}{})\sim_{G}(\lev{P'}{},\lev{F'}{}) \iff \overline{\pi}((\lev{P}{},\lev{F}{}))=\overline{\pi}((\lev{P'}{},\lev{F'}{})) 
\]
which is coherent with the procedure on the case of partial boundaries.

Observe that, defined on that way, $\overline{\pi}$ is not well defined in general, as it is not necessarily true that $\hat{\pi}(\lev{P}{})\sim_{S}\check{\pi}(\lev{F}{})$ (see Example \ref{ex:exe2}). However, we can overcome this problem under the assumptions of Propositions \ref{prop:proypares} and \ref{prop:levpares}.

\begin{prop}\label{prop:welldefpinotame}
	If we assume that the points $(P,F)\in \overline{M}$ with $P\neq\emptyset\neq F$ has lifts in $\overline{V}$ (see Prop. \ref{prop:levpares}) and that $M$ does not admit sequences with (future or past) divergent lifts, then the projection $\overline{\pi}$ is well defined and surjective.
\end{prop}
\begin{proof}
	Let us begin by showing that $\overline{\pi}$ is well defined. Take $(\lev{P}{},\lev{F}{})\in \overline{V}$ an arbitrary point and let us consider $\overline{\pi}((\lev{P}{},\lev{F}{}))$. If we are in the first case of \eqref{def:proy}, then \[\overline{\pi}((\lev{P}{},\lev{F}{}))=(\hat{\pi}(\lev{P'}{}),\check{\pi}(\lev{F'}{}))\]
	for some $\lev{P'}{}\neq \emptyset\neq \lev{F'}{}$. Prop. \ref{prop:proypares} ensures that $(\hat{\pi}(\lev{P'}{}),\check{\pi}(\lev{F'}{}))\in \overline{M}$ and the projection is well defined.
	
	Now, let us assume that we are on the second case of \eqref{def:proy}, and then, that no point $(\lev{P'}{},\lev{F'}{})$ with both components non empty can be $\sim_{G}$
	related with $(\lev{P}{},\lev{F}{})$. In particular, one of the components of the point should be empty, say $\lev{F}{}=\emptyset$ (the other case is analogous). In this case, $\overline{\pi}(\lev{P}{},\emptyset)=(P,\emptyset)$, and so, we have to prove that $P\sim_{S} \emptyset$. Otherwise, from the completeness of the c-completion (recall Thm. \ref{teo:propcomple} (iii)), the terminal set $P$ should be $S$-related with a terminal set $F\neq\emptyset$, determining the point $(P,F)\in \overline{M}$. From Prop. \ref{prop:levpares}, there are non empty lifts $\lev{P'}{}$ and $\lev{F'}{}$ of $P$ and $F$ such that $(\lev{P'}{},\lev{F'}{})\in \overline{V}$. However, we have that $\hat{\pi}(\lev{P}{})=\hat{\pi}(\lev{P'}{})$, and so, that $(\lev{P}{},\emptyset)\sim_{G}(\lev{P'}{},\lev{F'}{})$, a contradiction. In conclusion, $P\sim_{S}\emptyset$ and the projection is also well defined in this case.
	
	For the surjectivity, consider $(P,F)\in \overline{M}$. If $(P,F)$ has both components non empty, then by hypothesis admits a lift on $\overline{V}$ which projects on it. Otherwise, assume that $F=\emptyset$ (the other case is analogous) and take $\lev{P}{}$ any lift of $P$. From completeness of the c-completion, there exists $\lev{F}{}$ such that $(\lev{P}{},\lev{F}{})\in \overline{V}$. Moreover $\lev{F}{}$ has to be empty as, otherwise, recalling that $\overline{\pi}$ is well defined, $P\sim_{S}\check{\pi}(\lev{F}{})$ (which is not possible as $P\sim_{S} \emptyset$). Hence, any point $(\lev{P}{},\lev{F}{})\in \overline{V}$ with $\hat{\pi}(\lev{P}{})=P$ has $\lev{F}{}=\emptyset$ and, from the definition of $\overline{\pi}$, we deduce that $\overline{\pi}((\lev{P}{},\emptyset))=(P,\emptyset)$, as desired.
	
\end{proof} 

If we assume that $\pi$ is tame, the definition of $\overline{\pi}$ is simplified significantly due to the following result:
\begin{lemma}\label{wellprojectedtame}
	Assume that the projection $\pi:V\rightarrow M$ is tame. If $(\lev{P},\emptyset)\sim_{G} (\lev{P'}{},\lev{F'}{})$, then $\lev{F'}{}=\emptyset$ (and analogously for the case $(\emptyset,\lev{F}{})$).
\end{lemma}
\begin{proof}
	The result is straightforward, once we recall that in tame projections, if $\hat{\pi}(\lev{P}{})=\hat{\pi}(\lev{P'}{})$, then there exists $g\in G$ such that $\lev{P'}{}=g\,\lev{P}{}$. Therefore, if $\lev{F'}{}\neq\emptyset$ and $\lev{P'}{}\sim_{S}\lev{F'}{}$, it follows that $\lev{P}{}=g^{-1}\,\lev{P'}{}\sim_{S} g^{-1}\,\lev{F'}{}$, in contradiction with $\lev{P}{}\sim_{S}\emptyset$.
\end{proof}

Therefore, on the case of tame projection, the definition of $\overline{\pi}$ becomes  

\begin{equation}\label{eq:explpi} 
	\overline{\pi}((\lev{P}{},\lev{F}{}))=(\hat{\pi}(\lev{P}{}),\check{\pi}(\lev{F}{}))
\end{equation}

\smallskip 

In any case, considering $\pi$ tame or not, when the map $\overline{\pi}$ is well defined and surjective, we can proceed in complete analogy with the partial cases and obtain the following diagram:

\begin{center} 
	\begin{tikzpicture}
	\matrix (m) [matrix of math nodes,row sep=3em,column sep=4em,minimum width=2em]
	{
		\overline{V} & \\
		\overline{V}/G & \overline{M} \\};
	\path[-stealth]
	(m-1-1) edge node [left] {$\overline{\i}$} (m-2-1)
	(m-2-1.east|-m-2-2) edge node [below] {$\overline{\j}$} (m-2-2)
	(m-1-1) edge node [above] {$\overline{\pi}$} (m-2-2);
	\end{tikzpicture}
\end{center}
where two points in $\overline{V}$ are $G$-related if they project by $\overline{\pi}$ into the same point of $\overline{M}$.  

\subsection{At the chronological level}\label{sec:chronology}
Let us now study how is the behaviour of $\overline{\jmath}$ regarding the causal structure. As a first step, we need to define first a chronological relation on $\overline{V}/G$. For this, we will follow an approach inspired from \cite[Section 6.2]{FHSIso2}, where two equivalence classes $\overline{\i}((\lev{P}{},\lev{F}{})),\overline{\i}((\lev{P'}{},\lev{F'}{}))\in\overline{V}/G$ are chronologically related, $\overline{\i}((\lev{P}{},\lev{F}{}))\ll \overline{\i}((\lev{P'}{},\lev{F'}{}))$ if there exist $(\lev{P}{0},\lev{F}{0})\in \overline{\i}((\lev{P}{},\lev{F}{}))$ and $(\lev{P'}{0},\lev{F'}{0})\in \overline{\i}((\lev{P'}{},\lev{F'}{}))$ with $(\lev{P}{0},\lev{F}{0})\ll (\lev{P'}{0},\lev{F'}{0})$ in $\overline{V}$.

In general, and under the hypothesis that $\overline{\pi}$ is well defined and surjective, we can obtain that both spaces inherits the same causal structure.

\begin{prop}\label{prop:chronologicallevel}
	Let $\pi:V\rightarrow M$ a spacetime covering projection and assume that $\overline{\pi}$ is well defined and surjective. Denote by $\overline{\j}$ the corresponding map between $\overline{V}/G$ and $\overline{M}$. Then, the bijection $\overline{\j}$ is a chronological isomorphism, that is,
	
	\[
	(P,F)\ll (P',F') \iff \overline{\j}^{-1}((P,F))\ll \overline{\j}^{-1}((P',F'))
	\]
	 
\end{prop}

\begin{proof}
	Let us start by fixing some notation. Consider $(P,F),(P',F')\in \overline{M}$ and denote by $(\lev{P}{},\lev{F}{}),(\lev{P'}{},\lev{F'}{})\in \overline{V}$ two corresponding lifts. It follows that $\overline{\j}^{-1}((P,F))=\overline{\i}((\lev{P}{},\lev{F}{}))$ and $\overline{\j}^{-1}((P',F'))=\overline{\i}((\lev{P'}{},\lev{F'}{}))$.
	
	Assume that $\overline{\i}((\lev{P}{},\lev{F}{}))\ll \overline{\i}((\lev{P'}{},\lev{F'}{}))$ and, without loss of generality, that $(\lev{P}{},\lev{F}{})\ll (\lev{P'}{},\lev{F'}{})$. Then, $\lev{F}{}\cap \lev{P'}{}\neq\emptyset$ and from the first bullet point of Prop. \ref{prop:caracrel}, that $F\cap P'\neq\emptyset$. Therefore, $(P,F)\ll (P',F')$ and the left implication follows.
	
	For the other implication, assume that $(P,F)\ll (P',F')$, i.e., $F\cap P'\neq\emptyset$ and let $x\in F\cap P'$. As $x\in F$ and $\check{\pi}(\lev{F}{})=F$, Prop. \ref{prop:caracrel} ensures that there exists a point $p\in V$ with $\pi(p)=x$ such that $p\in \lev{F}{}$. Reasoning in the same way but fixing this lifted $p\in V$ of $x$, we can show that there exists $g\in G$ such that $p\in g\,\lev{P'}{}$ (recall that $\hat{\pi}(\lev{P'}{})=P'$). In conclusion,  $p\in \lev{F}{}\cap g\,\lev{P'}{}$ and so $(\lev{P}{},\lev{F}{})\ll (g\,\lev{P'}{},g\,\lev{F'}{})$. Hence, $\overline{\i}((\lev{P}{},\lev{F}{}))\ll \overline{\i}((\lev{P'}{},\lev{F'}{}))(=\overline{\i}((g\,\lev{P'}{},g\,\lev{F'}{})))$ and the right implication follows.
\end{proof}

\subsection{At the topological level}\label{sec:totaltopo} 
Finally, in this section we will compare the topological structures of both, $\overline{V}/G$ and $\overline{M}$. Let us start by fixing some notation. $\overline{M}$ and $\overline{V}$ will be endowed with the corresponding chronological topology, while $\overline{V}/G$ will be with the induced quotient topology from $\overline{V}$. In concordance with Section \ref{sec:Partial}, we will denote by $L_{M}$ the chronological limit on $\overline{M}$, by $L_{V}$ the chronological limit on $\overline{V}$ and by $L_G$ the quotient limit operator on $\overline{V}/G$ induced from $L_{V}$ (recall equation \eqref{deflimitQ}). 

In spite of the partial cases where the openness of the map $\hat{\jmath}$ is always ensured,  in general the map $\overline{\jmath}$ is neither continuous nor open. In fact, the following result summarize the only cases where $\overline{\jmath}$ is well behaved in general respect the limit operator.

\begin{prop}\label{prop:topgeneral}
	Let $\pi:V\rightarrow M$ be a spacetime covering and assume that $\overline{\pi}$ is well defined and surjective. Then
	
	\begin{itemize}
		\item[(a)] If $(\lev{P}{},\lev{F}{})\in L_{V}(\lev{\sigma}{})$ for some sequence $\lev{\sigma}{}\subset \overline{V}$ with $\lev{P}{}\neq\emptyset\neq \lev{F}{}$, then $(P,F)\in L_{M}(\sigma)$, where $(P,F)=\overline{\pi}((\lev{P}{},\lev{F}{}))$ and $\sigma=\overline{\pi}(\lev{\sigma}{})$.
		
		\item[(b)] If $(P,\emptyset)\in L_M(\sigma)$ (analogously for $(\emptyset,F)\in L_{M}(\sigma)$) for some sequence $\sigma\subset \overline{M}$, then there exist a subsequence $\kappa\subset \sigma$ and lifts $(\lev{P}{},\emptyset)$ and $\lev{\kappa}{}$ of $(P,\emptyset)$ and $\kappa$ respectively such that $(\lev{P}{},\emptyset)\in L_{V}(\lev{\kappa}{})$.
	\end{itemize}
\end{prop} 
\begin{proof}
	Assertion (b) is a direct consequence of \eqref{eq:deflimcrono}, Rem. \ref{rem:aux1} and the fact that any lift $(\lev{P}{},\lev{F}{})\in \overline{\pi}^{-1}((P,\emptyset))$ has $\lev{F}{}=\emptyset$, so let us focus on assertion (a). For this, recall that from the definition of the chronological limit, $\lev{P}{}\subset {\rm LI}(\{\lev{P}{n}\}_{n})$ and $\lev{F}{}\subset {\rm LI}(\{\lev{F}{n}\}_{n})$. As the projection is well behaved with contentions, we have that $P\subset {\rm LI}(\{P_n\}_{n})$ and $F \subset {\rm LI}(\{F_n\}_{n})$, which is enough to ensure that $(P,F)\in L_{M}(\{(P_n,F_n)\}_{n})$ (recall \eqref{eq:limitsimple}).
	
\end{proof}

The other cases (that is, when $(\lev{P}{},\lev{F}{})$ has one empty component or when $P\neq\emptyset\neq F$) are false in general, as it is proved by Examples \ref{ex:exemHarris} and \ref{ex:exe3}. On the first one there exists a sequence $\{q_n\}_{n}\subset V$ converging to a point of the form $(\lev{P}{},\emptyset)$, while its projection converges to a point $(P',\emptyset)$ with $\hat{\pi}(\lev{P}{})=P\subsetneq P'$. On the second example, the sequence $\{x_n\}_{n}$ converges to $(P,F)$ in $\overline{M}$, however $\{x_n\}_{n}$ has no convergent lift on the corresponding $\overline{V}$.

The first case is directly related with the non continuity of $\hat{\jmath}$. In fact, we can easily prove that:

\begin{prop}\label{prop:contparimtot}
	Let $\pi:V\rightarrow M$ a spacetime covering with $\overline{\pi}$ well defined and surjective. If $\overline{\pi}(\lev{P}{},\emptyset)=(P,\emptyset)$, $\overline{\pi}(\emptyset,\lev{F}{})=(\emptyset,F)$ and $M$ has no sequence with divergent lifts, the map $\overline{\pi}$ (and so, $\overline{\j}$) is continuous.
\end{prop}
  
  \begin{proof}
  For the continuity of $\overline{\jmath}$ is enough to show that, given a point $(\lev{P}{},\lev{F}{})\in \overline{V}$ and a sequence $\{(\lev{P}{n},\lev{F}{n})\}_{n} \subset \overline{V}$ with $(\lev{P}{},\lev{F}{})\in L_{V}(\{(\lev{P}{n},\lev{F}{n})\}_{n})$, then $(P,F)\in L_{M}(\{(P_n,F_n)\}_{n})$, where $(P,F)=\overline{\pi}(\lev{P}{},\lev{F}{})$ and $(P_n,F_n)=\overline{\pi}(\lev{P}{n},\lev{F}{n})$. If $\lev{P}{}\neq\emptyset\neq \lev{F}{}$, the result follows from Prop. \ref{prop:topgeneral} (a). If $\lev{F}{}=\emptyset$ (the other case is analogous) we have that $\lev{P}{}\in \hat{L}_{V}(\{\lev{P}{n}\}_{n})$ and so, from Prop. \ref{thm:hatjcont}, that $P\in \hat{L}_{M}(\{P_n\}_{n})$. Finally, from hypothesis, $\overline{\pi}(\lev{P}{},\emptyset)=(P,\emptyset)\in\overline{M}$, so $(P,\emptyset)\in L_{M}(\{(P_n,F_n)\}_{n})$.
  \end{proof}

\begin{rem}\label{rem:aux3}
	It is important to note that both $\overline{\pi}(\lev{P}{},\emptyset)=(P,\emptyset)$ and $\overline{\pi}(\emptyset,\lev{F}{})=(\emptyset,F)$ follow when $\pi$ is (future and past) tame, as it was proved by Lemma \ref{wellprojectedtame} and \eqref{def:proy}.
\end{rem}

Let us give a closer look to the previous proof. Observe that the non existence of divergent lifts is used precisely when we deal with limit points of the form $(\lev{P}{},\emptyset)$ or $(\emptyset,\lev{F}{})$. Therefore, and recalling that the existence of divergent lifts  can occur only when a limit terminal set is contained in other (bigger) terminal set (see Def. \ref{def:divlif}), we deduce that the continuity of $\hat{\jmath}$ and $\check{\jmath}$ is not necessary when the boundary on $M$ is formed only by timelike and spatial points (see also Corollary \ref{cor:Harrisresult}). Hence, 

\begin{cor}\label{prop:contnolight}
	Let $\pi:V\rightarrow M$ be a projection satisfying: (i) $\overline{\pi}$ is well defined and surjective, (ii) $\overline{\pi}(\lev{P}{},\emptyset)=(P,\emptyset)$ and $\overline{\pi}(\emptyset,\lev{F}{})=(\emptyset,F)$; and (iii) $\overline{M}$ has no lightlike boundary points. Then, the map $\overline{\jmath}$ is continuous. 
\end{cor}

\smallskip

As we have mention at the beginning of the section, and in spite of the continuity, the openness of the partial maps $\hat{\jmath}$ and $\check{\jmath}$ is not enough to ensure the openness of $\overline{\jmath}$, as we can see on Example \ref{ex:exe3}. This means that an additional condition has to be imposed to obtain such an openness. In this sense, we will consider the condition of {\em finite chronology} whose properties will be studied in the following section.

\subsection{Group actions with the finite chronology property} \label{sec:totalfinite}
First of all, let us introduce the definition of finite chronology.
	
	\begin{defi}\label{def:finitelyachronal}
		Let $V$ be a spacetime and $G$ a group of isometries. We will say that the pair $(V,G)$ is {\em finitely chronological} if given two points $p,q\in V$ with $p\ll q$, there exist only a finite number of elements $g\in G$ such that $p\ll g\,q$. 
	\end{defi}
	
	The finite chronology property will be enough to ensure the openness of $\overline{\jmath}$  and it will also simplify the conditions to ensure when the map $\overline{\pi}$ is well defined and surjective. However, such a condition will not be enough to prove the continuity of $\hat{\jmath}$ or $\check{\jmath}$, as it is showed by Example \ref{ex:exe2}. Let us begin with a crucial lemma:

\begin{lemma}\label{lem:sequence}
	Assume that $(V,G)$ is finitely chronological and consider a point $p\in V$, a past-directed (resp. future-directed) chronological chain $\{p_n\}_{n} \subset V$ and a sequence $\{g_n\}_{n} \subset G$. If for all $n\in \mathbb{N}$, $p\ll g_n\,p_n$ (resp. $g_n\,p_n\ll p$), then there exists $n_0\in \mathbb{N}$ and a finite family $\{h_1,\dots,h_r\}\subset G$ such that for $n\geq n_0$, $g_n=h_i$ for some $i=1,\dots,r$. In fact, $n_0$ can be taken in such a way that $\{h_1,\dots,h_r\}\subset G(p,\{p_n\}_{n})$ (resp. $\{h_1,\dots,h_r\}\subset G(\{p_n\}_{n},p)$, where
	
	\[\begin{array}{c}
	G(p,\{p_n\}_{n}):=\{g\in G: p\ll g\,p_n\hbox{ for all }n\}\vspace{0.2cm}\\ 
	\left(G(\{p_n\}_{n},p):=\{g\in G: g\,p_n\ll p\hbox{ for all }n\} \right)
	\end{array}\]
	is a non empty finite set.
	
\end{lemma}

\begin{proof}
	The proof follows essentially by recalling that, for a fixed $k_0\in \mathbb{N}$ and $n\geq k_0$,
	
	\begin{equation}\label{eq:aux1}
	p\ll g_{n}\,p_{n}\ll g_{n}\,p_{k_0}.
	\end{equation}
	
	In particular, as there exist a finite number of elements $g\in G$ such that $p\ll g\,p_{k_0}$, $g_n$ should belong to a finite family of elements in $G$ for $n$ big enough. Moreover, we can take $\{h_1\dots,h_r\}\subset G$ such that, for all $h_i$, there exists a subsequence $\{g_{n^i_k}\}_{k}$ with $g_{n_k^i}=h_i$. In particular, there exists $n_0$ such that for each $n\geq n_0$ there exists $i(\equiv i(n))$ with $g_n=h_i$.
	
	For the second assertion, recall that the set $G(p,\{p_n\}_{n})$ is finite by the finitely chronological property. Moreover for each previous $h_i$, we know from \eqref{eq:aux1} that $p\ll h_i\,p_n$ for all $n$. Therefore, $\{h_1,\dots,h_r\}\subset G(p,\{p_n\}_{n})$ as desired.
\end{proof}

If we consider two points $p,p'\in V$ with $p\ll p'$, then it follows that $G(p',\{p_n\}_{n})\subseteq G(p,\{p_n\}_{n})$. This relation allow us to prove that the lifts of terminal sets are well behaved, at least when $(V,G)$ is finitely chronological, with respect the future and common pasts. Concretely,

\begin{lemma}\label{lem:pairachronal}
Consider $P\subset \downarrow F$ and take $\lev{P}{}, \lev{F}{}$ the corresponding lifts. If $(V,G)$ is finitely chronological then the set $G(\lev{P}{},\lev{F}{})$ defined by

\begin{equation}
G(\lev{P}{},\lev{F}{})=\{g\in G:\lev{P}{}\subset\downarrow g\,\lev{F}{}\}
\end{equation}
is non empty and finite.	
\end{lemma}

\begin{proof}

As a first step, we are going to characterize the set $G(\lev{P}{},\lev{F}{})$ in terms of the sequences defining $\lev{P}{}$ and $\lev{F}{}$. In this sense, let $\{x_n\}_{n}$ and $\{y_n\}_{n}$  be chronological sequences defining $P$ and $F$ resp., and $\{p_n\}_{n}, \{q_n\}_{n}$ the corresponding chronological lifts defining $\lev{P}{}$ and $\lev{F}{}$. Observe that the following chain of equivalences follow

\[
\begin{array}{rl}
g\in G(\lev{P}{},\lev{F}{}) & \iff \lev{P}{}\subset \downarrow g\, \lev{F}{}\\
 & \iff p_n\in g\,\lev{F}{} \hbox{ for all $n\in \mathbb{N}$}\\
 & \iff p_n\ll g\, q_m \hbox{ for all $n,m\in \mathbb{N}$}\\
 & \iff g\in G(p_n,\{q_m\}_{m}) \hbox{ for all $n\in \mathbb{N}$}
\end{array} 
\]

In particular, 

\begin{equation}\label{eq:aux2}
G(\lev{P}{},\lev{F}{})=\cap_{n\in \mathbb{N}} G(p_n,\{q_m\}_{m}).
\end{equation}

As a second step, recall that from hypothesis $P\subset \downarrow F$, and so, $x_n\ll y_m$ for all $n,m\in \mathbb{N}$. Hence, Prop. \ref{prop:caracrel} ensures that there exists a sequence $\{g_m\}_{m} \subset G$ such that $p_n \ll g_m\,q_m$ and so, from Lemma \ref{lem:sequence}, $G(p_n,\{q_m\}_{n})$ is non empty and finite for all $n$.

Then, $G(\lev{P}{},\lev{F}{})$ is the intersection of a numerable family of non empty and finite sets ordered by $G(p_{n+1},\{q_m\}_{m})\subset G(p_{n},\{q_m\}_{m})$. Therefore, it is a non empty and finite set.
\end{proof}

In particular, and as a consequence of previous Lemma and Props. \ref{prop:levpares} and \ref{prop:welldefpinotame}, we have that:

\begin{cor}\label{cor:welldefinedachronal}
	Let $\pi:V\rightarrow M$ be a spacetime covering with $(V,G)$ finitely chronological and assume that $M$ does not admit sequences with (future or past) divergent lifts. Then, the map $\overline{\pi}:\overline{V}\rightarrow \overline{M}$ is well defined and surjective.
\end{cor}

At this point a natural question arise at the point set level: is there any relation between $\overline{\pi}^{-1}((P,F))$ and the set $G(\lev{P}{},\lev{F}{})$. Intuitively, one can expect that for a fixed lift $\lev{P}{}$, the set $G(\lev{P}{},\lev{F}{})$ determines all the pairs of the form $(\lev{P}{},g\,\lev{F}{})\in \overline{V}$ with projection $(P,F)$. However, as we recall in Rem. \ref{rem:conjuntosSrelated}, it is not clear that, in general, all the lifts preserving the relation with the common future (or past) are $S$-related. Again, the finite chronology condition will be enough for this, as we will see on Lemma \ref{lem:lem4.12}. In order to prove such a lemma, we need first the following technical result:

\begin{lemma}\label{lem:aux1} 
	Let $\lev{P}{},\lev{P'}{}\in \hat{V}$ (resp. $\lev{F}{},\lev{F'}{}\in \check{V}$) be two points of the future (past) causal completion projecting to the same set $P\in \hat{M}$ ($F\in\check{M}$). Suppose one of the following situations:
	\begin{itemize}
		\item[(H1)] $(V,G)$ is finitely chronological and there exists $p\in V$ such that $\lev{P}{},\lev{P'}{}\subset I^-(p)$ ($\lev{F}{},\lev{F'}{}\subset I^+(p)$).
		
		\item[(H2)] $G$ is finite.
		
	\end{itemize}
	Then there exists $h'$ such that $\lev{P}{}=h'\,\lev{P'}{}$ ($\lev{F}{}=h'\,\lev{F'}{}$).
	In particular, it follows that if $G$ is finite, the projection $\pi$ is future (past) tame.

\end{lemma}
\begin{proof}
	Let $\{p_n\}_{n},\{p'_n\}_{n}$ be future chronological chains defining $\lev{P}{}$ and $\lev{P'}{}$ resp. As both sets project onto the same $P$, it follows that the projection of such sequences $\{x_n\}_{n},\{x'_n\}_{n}$ generate $P$. In particular, for each $n$ there exists $m(n)$ big enough such that $x_n\ll x'_{m(n)}$. We will consider $\{m(n)\}_{n}$ a strictly increasing sequence, so $\{x'_{m(n)}\}_{n}$ is a subsequence of $\{x'_n\}_{n}$ and generates the same $P$ (and, accordingly, $\{p'_{m(n)}\}_{n}$ generates $\lev{P'}{}$). From Prop. \ref{prop:caracrel} it follows that there exists a sequence $\{g_n\}_{n} \subset G$ such that $g_n\,p_n\ll p'_{m(n)}$ for all $n$.
	
	Now observe that, in either situation (H1) nor (H2), and up to a subsequence, $\{g_n\}_{n}$ can be considered a constant sequence (say $g_n=h\in G$ for all $n$). In the case that $G$ is finite the argument is straightforward. In the other case, recall that from (H1) we have that $g_n\,p_n\ll p'_{m(n)}\ll p$, and so the assertion follows from Lemma \ref{lem:sequence}. Therefore, $h\,p_n\ll p'_{m(n)}$ for all $n$, and hence, $h\,\lev{P}{}\subset \lev{P'}{}$.
	By interchanging the roles of $\lev{P}{}$ and $\lev{P'}{}$ we find another $h'$ such that $h'\,\lev{P'}{}\subset \lev{P}{}$. 
	
	Now, we can join both contentions together in the following way 
	
	\begin{equation}\label{eq:aux3} 
	g\,\lev{P}{}\subset h'\,\lev{P'}{}\subset \lev{P}{} 
	\end{equation} 
	for $g=h'h$; and then construct the chain:
\[
	\lev{P}{}\supset g\,\lev{P}{}\supset\,(g)^2\,\lev{P}{}\supset\dots\supset (g)^{n}\,\lev{P}{}\supset \dots
	\] 
	where $(g)^i$ denotes the iteration of the action by $g$ $i$-times. Now observe that under the hypothesis of the lemma, there exists $i_0$ such that $(g)^{i_0}=e$. This assertion is again straightforward under the assumption of $G$ finite, so let us focus on the hypothesis (H1). If by contradiction $(g)^{i}\neq(g)^{j}$ for all $i\neq j$, and recalling that $\lev{P}{}\subset I^-(p)$, we deduce that $(g)^i\,\lev{P}{}\subset I^-(p)$ for all $i$. which contradicts that $(V,G)$ is finitely chronological (the point $p$ will be chronologically related with  $(g)^i\,q$ for any $q\in \lev{P}{}$ and $i\in \mathbb{N}$).
	
	Summarizing we deduce that $g\,\lev{P}{}=\lev{P}{}$ and from \eqref{eq:aux3} we obtain that $\lev{P}{}=h'\,\lev{P'}{}$, as desired.	
\end{proof}

\begin{rem}\label{rem:aux2}
	Observe that we have also proved in previous lemma that if $g\,\lev{P}{}\subset \lev{P}{}$ for some $g$ and, or $G$ is finite, or $(V,G)$ is finitely chronological and there exists $p\in V$ with $\lev{P}{}\subset I^-(p)$, then $g\, \lev{P}{}=\lev{P}{}$ (an analogous result for past sets follows). 
\end{rem}

\begin{lemma}\label{lem:lem4.12}
	Assume that $(V,G)$ is finitely chronological. If $\lev{P}{}$ and $\lev{F}{}$ are terminal sets with $\check{\pi}(\lev{P}{})=P\sim_S F=\hat{\pi}(\lev{F}{})$, then  $\lev{P}{}\sim_S g\,\lev{F}{}$ for all $g\in G(\lev{P}{},\lev{F}{})$.
\end{lemma}

\begin{proof}
	 Assume without loss of generality that $e\in G(\lev{P}{},\lev{F}{})$, and so, that $\lev{F}{}\subset \uparrow \lev{P}{}$. By contradiction, let us assume that $\lev{P}{}$ is not $S$-related with $\lev{F}{}$. Recalling Rem. \ref{rem:conjuntosSrelated}, we ensure the existence of a terminal set $\lev{F'}{}$ with $\lev{F}{}\subsetneq \lev{F'}{}\subset\uparrow \lev{P}{}$ and satisfying that $\check{\pi}(\lev{F}{})=\check{\pi}(\lev{F'}{})$. As $(V,G)$ is finite chronological and there exists $p\in V$ such that $\lev{F'}{},\lev{F}{}\subset I^+(p)$ (take any $p\in \lev{P}{}$), Lemma \ref{lem:aux1} ensures that there exist $h\in G$ such that $\lev{F'}{}=h\,\lev{F}{}$. But then, recalling Rem \ref{rem:aux2}, we arrive to a contradiction with $\lev{F}{}\subsetneq \lev{F'}{}=h\,\lev{F}{}$.
\end{proof}

With all previous machinery set, we are now in conditions to prove the openness of $\overline{\jmath}$ under the assumption of finite chronology:

\begin{prop}\label{prop:jopen}
	Let $\pi:V\rightarrow M$ be spacetime covering projection with $(V,G)$ finitely chronological and assume that $\overline{\pi}$ is well defined and surjective. Then, the map $\overline{\pi}$ induces an open map $\overline{\j}$ from $\overline{V}/G$ to $\overline{M}$.
\end{prop}

\begin{proof}
	Let $\{(P_n,F_n)\}_{n}\subset \overline{M}$ be a sequence and $(P,F)\in \overline{M}$ a point such that $(P,F)\in L_{M}(\{(P_n,F_n)\}_{n})$.  Our aim is to show that, up to a subsequence, $(P_n,F_n)$ and $(P,F)$ admit lifts $(\lev{P'}{n},\lev{F'}{n})$ and $(\lev{P'}{},\lev{F'}{})$ with  $(\lev{P'}{},\lev{F'}{})\in L_{V}(\{(\lev{P'}{n},\lev{F'}{n})\}_{n})$, and hence, that $\overline{\j}^{-1}((P,F))\in L_G(\{\overline{\j}^{-1}(P_n,F_n)\}_{n})$ (recall \eqref{deflimitQ}). 
	Observe that the case where $F$ or $P$ is empty follows from Prop. \ref{prop:topgeneral} (b), so we only need to focus on the case where both sets are non empty.
	
Assume that $P\neq\emptyset\neq F$ and let $\lev{P}{},\lev{F}{},\lev{P}{n},\lev{F}{n}$ be some fixed lifts of $P,F,P_n,F_n$ respectively. Consider $\{x_n\}_{n}$ and $\{y_n\}_{n}$ chronological sequences defining $P$ and $F$ and, as usual, denote by $\{p_n\}_{n}$ and $\{q_n\}_{n}$ the corresponding lifts defining $\lev{P}{}$ and $\lev{F}{}$. Let us denote by $\{m(n)\}_{n}$ a sequence in $\mathbb{N}$ with $m(n+1)\geq m(n)+1$ and satisfying that $x_n\subset P_{m(n)}$ and $y_n\in F_{m(n)}$. Now, as $x_n\in P_{m(n)}$, Prop. \ref{prop:caracrel} ensures that $p_n\in g_n\,\lev{P}{m(n)}$ for some $g_n\in G$. From Lemma \ref{lem:pairachronal}, we know that the set $G(g_n\,\lev{P}{m(n)},\lev{F}{m(n)})$ is non empty and, from Lemma \ref{lem:lem4.12}, that for any $g'_{n}\in G(g_n\,\lev{P}{m(n)},\lev{F}{m(n)})$,  $g_n\,\lev{P}{m(n)}\sim_S g'_n\,\lev{F}{m(n)}$. Finally, again from Prop. \ref{prop:caracrel} and $y_n\in F_{m(n)}$, there exists $h_n\in G$ such that $h_n\,q_n\in g'_n\,\lev{F}{m(n)}$. 
	
Now, let us observe that from $g_n\,\lev{P}{m(n)}\subset \downarrow g'_n\,\lev{F}{m(n)}$, it follows that $p_n\ll h_n\,q_n$. In particular, we have the chain
	
	\[
	p_1\ll p_n\ll h_n q_n
	\]
	and then, from Lemma \ref{lem:sequence}, we can ensure that, up to a subsequence, $\{h_n\}_{n}$ is constant, say $h_n=h\in G$ for all $n$. In particular, for any $i$ and all $n>i$, it follows that 
	\[
	p_i\ll p_n\ll h\,q_n.
	\]
	
	In particular, $\lev{P}{}\subset \downarrow h\,\lev{F}{}$ and so $h\in G(\lev{P}{},\lev{F}{})$. Hence, Lemma \ref{lem:lem4.12} ensures that both sets $\lev{P}{}$ and $h\,\lev{F}{}$ are S-related. 
	
Summarizing:
	
	\begin{itemize}
		\item The pairs $(\lev{P}{}, h\, \lev{F}{})$ and $(g_n\,\lev{P}{m(n)},g'_n\,\lev{F}{m(n)})$ belongs to $\overline{V}$.
		\item[]
		\item $\lev{P}{}\subset {\rm LI}(\{g_{n}\lev{P}{m(n)}\}_{n})$ and $h\,\lev{F}{}\subset {\rm LI}(\{g'_n\,\lev{F}{m(n)}\}_{n})$, thus \[(\lev{P}{},h\,\lev{F}{})\in L_V(\{(g_n\lev{P}{m(n)},g'_n\lev{F}{m(n)})\}_{n}).\]
	\end{itemize}
	
	In conclusion, and always up to a subsequence, if $(P,F)\in L_{M}(\{(P_n,F_n)\}_{n})$ we can always obtain appropriate lifts $(\lev{P'}{},\lev{F'}{})$ and $\{(\lev{P'}{n},\lev{F'}{n})\}_{n}$ such that $(\lev{P'}{},\lev{F'}{})\in L_{V}(\{(\lev{P'}{n},\lev{F'}{n})\}_{n})$. The result follows then as a consequence of Prop. \ref{prop:limcont} applied to $\overline{\jmath}^{-1}$.
	
\end{proof}

\smallskip

As a final remark of this section, we will show how finite chronology let us simplify some of our previous hypothesis for the well definition and continuity of $\overline{\jmath}$. 
 In fact, the condition of $M$ having no sequence with divergent lifts (which is almost equivalent to the continuity of $\hat{\pi}$ and $\check{\pi}$, recall Prop. \ref{thm:hatjcont}) imposed in Prop. \ref{prop:welldefpinotame} can be substituted by a topological requirement on $\overline{V}$:
	
	\begin{cor}\label{prop:welldefHaus}
		Assume that $(V,G)$ is finitely chronological and that $\overline{V}$ is Hausdorff. Then, $\overline{\pi}$ is well defined and surjective.
	\end{cor}
	\begin{proof}   
		We only need to show, according to Prop. \ref{prop:proypares}, that any past-directed chronological chain on $M$ has no future divergent lifts (the other case will be completely analogous). Let $\{y_n\}_{n}$ be a past-directed chronological chain and consider $\{q_n\}_{n}$ a past chronological sequence in $V$ with $\pi(q_n)=y_n$. Suppose that there exist $\{h_n\}_{n},\{g_n\}_{n} \subset G$ and $\overline{P}{},\overline{P'}{}\in \hat{V}$ such that $\lev{P}{}\in \hat{L}_{V}(\{I^-(h_n\,q_n)\}_{n})$ and $\lev{P'}{}\in \hat{L}_{V}(\{I^-(g_n\,q_n)\}_{n})$.
		
		Take $p\in \lev{P}{}$. From $\lev{P}{}\in \hat{L}_{V}(\{I^-(h_n\,q_n)\}_{n})$ we have that $p\ll h_n\,q_n$ for $n$ big enough. As $(V,G)$ is finitely chronological, Lemma \ref{lem:sequence} ensures that, up to a subsequence, $h_n=h_0$ for some fixed $h_0\in G$. Reasoning in the same way with $\lev{P'}{}$ and $\{g_n\}_{n}$, we can ensure that $g_n=g_0$ for some fixed $g_0\in G$.
		
		Hence, we have that $\lev{P}{}\in \hat{L}_{V}(\{I^-(h_0\,q_n)\}_{n})$ and $\lev{P'}{}\in \hat{L}_{V}(\{I^-(g_0\,q_n)\}_{n})$. From the first inclusion, we deduce that $(g_0h^{-1}_0)\,\lev{P}{}\in \hat{L}_{V}(\{I^-(g_0\,q_n)\}_{n})$ and recalling both, the second inclusion and the hypothesis that $\overline{V}$ is Hausdorff, we ensure that $(g_0h^{-1}_0)\,\lev{P}{}=\lev{P'}{}$ (and so both sets project into the same set in $\hat{M}$). In conclusion, the sequence $\{y_n\}_{n}$ cannot admit future divergent lifts.   
	\end{proof}

At the topological level, we also have to impose some conditions on $\overline{M}$, obtaining:

\begin{cor}\label{cor:aplicacion}
	Assume that $(V,G)$ is finitely chronological, $\overline{V}$ it is Hausdorff and $\overline{M}$ has no lightlike points. Then, $\overline{V}/G\equiv \overline{M}$, i.e., both $\overline{V}/G$ and $\overline{M}$ are homeomorphic and chronologically isomorphic.
\end{cor}
\begin{proof}
From Cor. \ref{prop:welldefHaus} follows that $\overline{\pi}$ is well defined and surjective. Then, Prop. \ref{prop:chronologicallevel} and \ref{prop:jopen} ensure both, that $\overline{\jmath}$ is a chronological isomorphism and an open map.

Hence, it only rest to show that $\overline{\j}$ is continuous. But this follows from Cor. \ref{prop:contnolight}, recalling that Lemma \ref{lem:lem4.12} ensures that $\overline{\pi}(\lev{P}{},\emptyset)=(P,\emptyset)$ and $\overline{\pi}(\emptyset,\lev{F}{})=(\emptyset,F)$.  
\end{proof}

Ideally, one would like to impose conditions only on $\overline{V}$ in order to ensure that $\overline{V}/G$ and 
$\overline{M}$ have the same structures. For example, and in the spirit of Cor. \ref{cor:aplicacion},
we would like to impose on $\overline{V}$ the non-existence of lightlike boundary points to obtain the non-existence of lightlike boundary points on $\overline{M}$, and therefore the continuity of $\overline{\jmath}$. However, the lack of lightlike points in $\overline{V}$ is not enough to ensure the same property on $\overline{M}$ (see Example \ref{ex:ex5}). Nevertheless the situation is very controlled and it is related again with the existence of very particular divergent lifts. In fact, we can prove that (compare with Prop. \ref{prop:new}):

\begin{lemma}\label{lem:lemnew} 
Let $\pi:V\rightarrow M$ be a spacetime projection. Assume that $\overline{\pi}$ is well-defined and it satisfies that $\overline{\pi}(\lev{P}{},\emptyset)=(P,\emptyset)$ and $\overline{\pi}(\emptyset,\lev{F}{})=(\emptyset,F)$. If $\overline{V}$ has no lightlike points and the $G$-orbits for both $\hat{V}$ and $\check{V}$ are closed (with the corresponding topologies), then $\overline{M}$ has no lightlike points.
\end{lemma}
\begin{proof}
	Assume by contradiction that $\overline{M}$ has lightlike points, that is, that there exists $(P,\emptyset)\in\overline{M}$ and $P'\in\hat{M}$ such that $P\subsetneq P'$ (the case with past sets will be analogous). Let $\{x_{n}\}_{n}$ and $\{x_{n}'\}_{n}$ be chronological chains generating $P$ and $P'$ resp. and consider $\lev{P}{}$, $\lev{P'}{}$, $\{p_{n}\}_{n}$ and $\{p_{n}'\}_{n}$ the corresponding lifts on $\hat{V}$. From hypothesis, it follows that $(\lev{P}{},\emptyset)\in \overline{V}$.
	
	As $P\subset P'$ we deduce that for all $n$ $x_n\ll x'_{n'}$ for $n'$ big enough, so Prop. \ref{prop:caracrel} ensures that there exists $g_n$ such that $p_n\ll g_n\,p'_n\in g_n\, \lev{P'}{}$. It follows then that $\lev{P}{}\subset {\rm LI}(\{g_n\, \lev{P'}{}\}_{n})$. Moreover, it also follows that $\bar{P}\in \hat{L}_{V}(\{g_n \,\bar{P'}\}_{n})$ as, otherwise, there exists $\bar{P''}$ such that $\bar{P}\subsetneq \bar{P''}$ and this is not possible as $\overline{V}$ has no lightlike points.
	
	Finally, and from the hypothesis that the $G$-orbits are closed on $\hat{V}$ with the future chronological topology, it follows that $\bar{P}\in \{g\, \lev{P'}{}\}_{g\in G}$, i.e., there exists $g_0  \in G$ such that $\bar{P}=g_0\,\bar{P'}$. In conclusion, and taking projections, we obtain that $P=P'$, a contradiction.	
\end{proof}

\medskip

As a consequence of both, Cor. \ref{cor:aplicacion} and Lemma \ref{lem:lemnew}, we obtain the following result:  

\begin{cor}
\label{cor:expansiblyhomeomorphism}
Assume that $(V,G)$ is finitely chronological, $\overline{V}$ satisfies that it is Hausdorff, has no lightlike points and the $G$-orbits in both, $\hat{V}$ and $\check{V}$ are closed. Then, $\overline{V}/G\equiv \overline{M}$, i.e., both $\overline{V}/G$ and $\overline{M}$ are homeomorphic and chronologically isomorphic.
\end{cor}


\subsection{Proof of Theorem \ref{thm:main2}}

At the point set level, the first assertion on (PS1) is proved on Prop. \ref{prop:welldefpinotame} and the second one follows from Lemma \ref{lem:pairachronal}. (PS2) is a consequence of Lemma \ref{wellprojectedtame} and (PS3) is proved in 
Cor. \ref{prop:welldefHaus}.

At the chronological level, assertion (CH) is proved on Prop. \ref{prop:chronologicallevel}.

At the topological level, (TP1) (i) is proved in Prop. \ref{prop:contparimtot} (see also Rem. \ref{rem:aux3}), (TP1) (ii) is Cor. \ref{prop:contnolight} and (TP2) follows from Prop. \ref{prop:jopen}. 

For the last paragraph, (a) follows from (PS1), (CH) and (TP1), while for (b) we have to consider (PS3) and Cor. \ref{cor:aplicacion} instead of (PS1) and (TP1). The last assertion (c), is proved in Cor. \ref{cor:expansiblyhomeomorphism}.

\section{On the optimality of the results: Some examples}\label{sec:examples}

Along this section, we will include some examples showing that our main results are optimal. It is worth pointing out that in all the examples $\# L_{M}(\sigma)$ will be bounded, and so, according to Lemma \ref{lem:firstUTS}, that $L_{M}$ will be of first order UTS. This is specially relevant recalling Prop. \ref{thm:hatjcont}, as it means that in all our examples the non existence of divergent lifts characterize the continuity of $\hat{\pi}$ and $\check{\pi}$. 

\smallskip 

Let us start with the example due to Harris where $\hat{\j}$ is not continuous. Here, we will include only the main properties of his example, referring the reader to \cite{H} for details.

\begin{exe}\label{ex:exemHarris}
\begin{figure}
	\centering
	\ifpdf
	\setlength{\unitlength}{1bp}%
	\begin{picture}(451.71, 281.81)(0,0)
	\put(0,0){\includegraphics{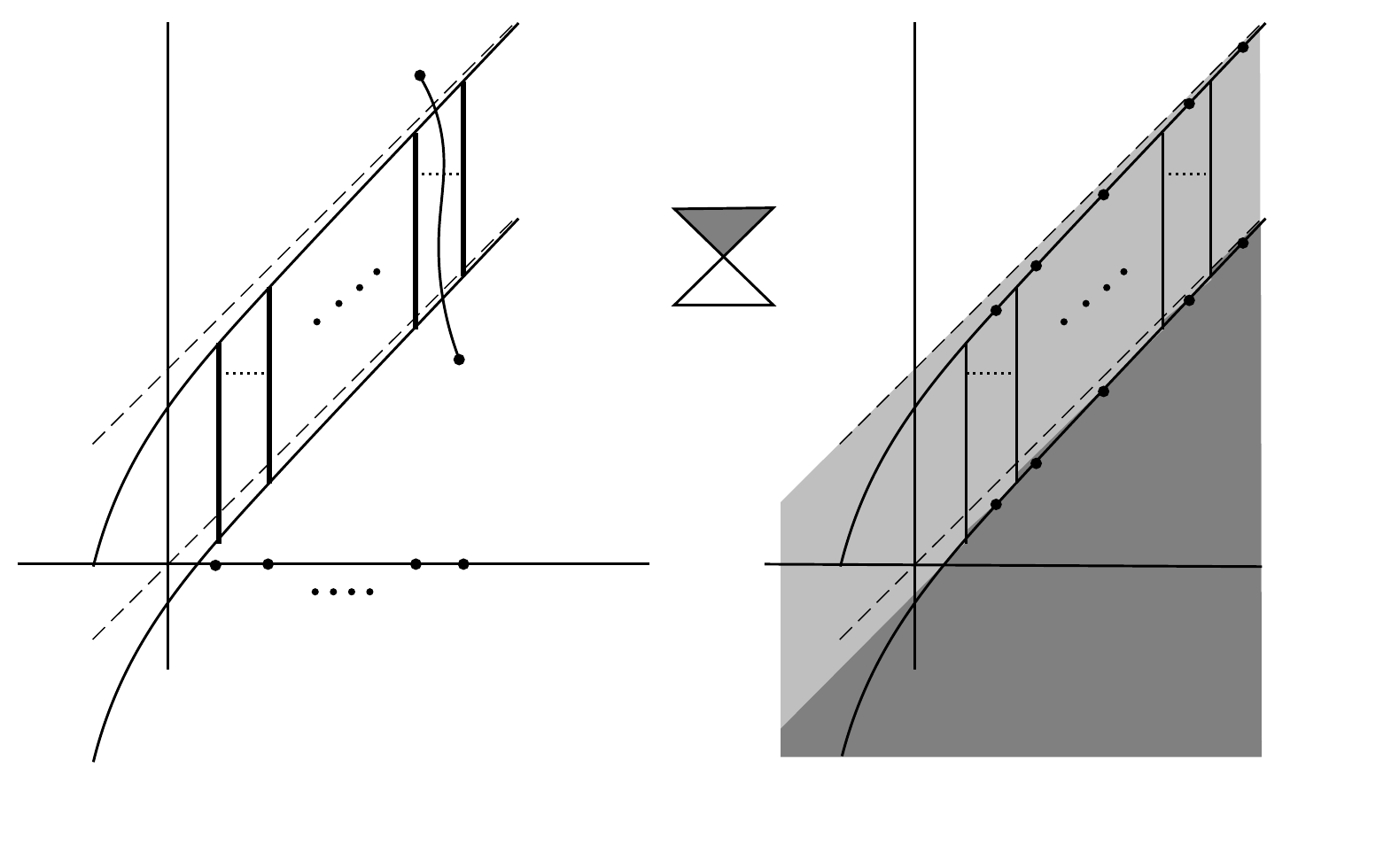}}
	\put(66.09,86.81){\fontsize{8.03}{9.64}\selectfont $1$}
	\put(82.77,86.81){\fontsize{8.03}{9.64}\selectfont $2$}
	\put(132.02,86.41){\fontsize{8.03}{9.64}\selectfont $n$}
	\put(146.32,86.81){\fontsize{8.03}{9.64}\selectfont $n+1$}
	\put(60.16,133.67){\fontsize{8.03}{9.64}\selectfont $S_1$}
	\put(88.29,146.78){\fontsize{8.03}{9.64}\selectfont $S_2$}
	\put(150.68,212.60){\fontsize{8.03}{9.64}\selectfont $S_{n+1}$}
	\put(150.99,226.41){\fontsize{8.03}{9.64}\selectfont $H_n$}
	\put(25.94,125.10){\fontsize{8.03}{9.64}\selectfont $\sigma_+$}
	\put(26.53,61.13){\fontsize{8.03}{9.64}\selectfont $\sigma_-$}
	\put(72.66,162.45){\fontsize{8.03}{9.64}\selectfont $H_1$}
	\put(150.73,160.64){\fontsize{6.42}{7.71}\selectfont $x$}
	\put(137.95,259.21){\fontsize{6.42}{7.71}\selectfont $y$}
	\put(145.64,178.09){\fontsize{8.03}{9.64}\selectfont $\gamma$}
	\put(322.02,109.99){\fontsize{8.03}{9.64}\selectfont $x_1$}
	\put(339.50,125.99){\fontsize{8.03}{9.64}\selectfont $x_2$}
	\put(389.25,178.11){\fontsize{8.03}{9.64}\selectfont $x_n$}
	\put(409.84,200.69){\fontsize{8.03}{9.64}\selectfont $x_{n+1}$}
	\put(310.48,187.20){\fontsize{8.03}{9.64}\selectfont $y_1$}
	\put(323.73,201.88){\fontsize{8.03}{9.64}\selectfont $y_2$}
	\put(373.40,253.58){\fontsize{8.03}{9.64}\selectfont $y_n$}
	\put(380.60,269.87){\fontsize{8.03}{9.64}\selectfont $y_{n+1}$}
	\put(95.72,9.13){\fontsize{16.06}{19.27}\selectfont (A)}
	\put(331.84,9.72){\fontsize{16.06}{19.27}\selectfont (B)}
	\put(123.50,186.73){\fontsize{8.03}{9.64}\selectfont $S_n$}
	\end{picture}%
	\else
	\setlength{\unitlength}{1bp}%
	\begin{picture}(451.71, 281.81)(0,0)
	\put(0,0){\includegraphics{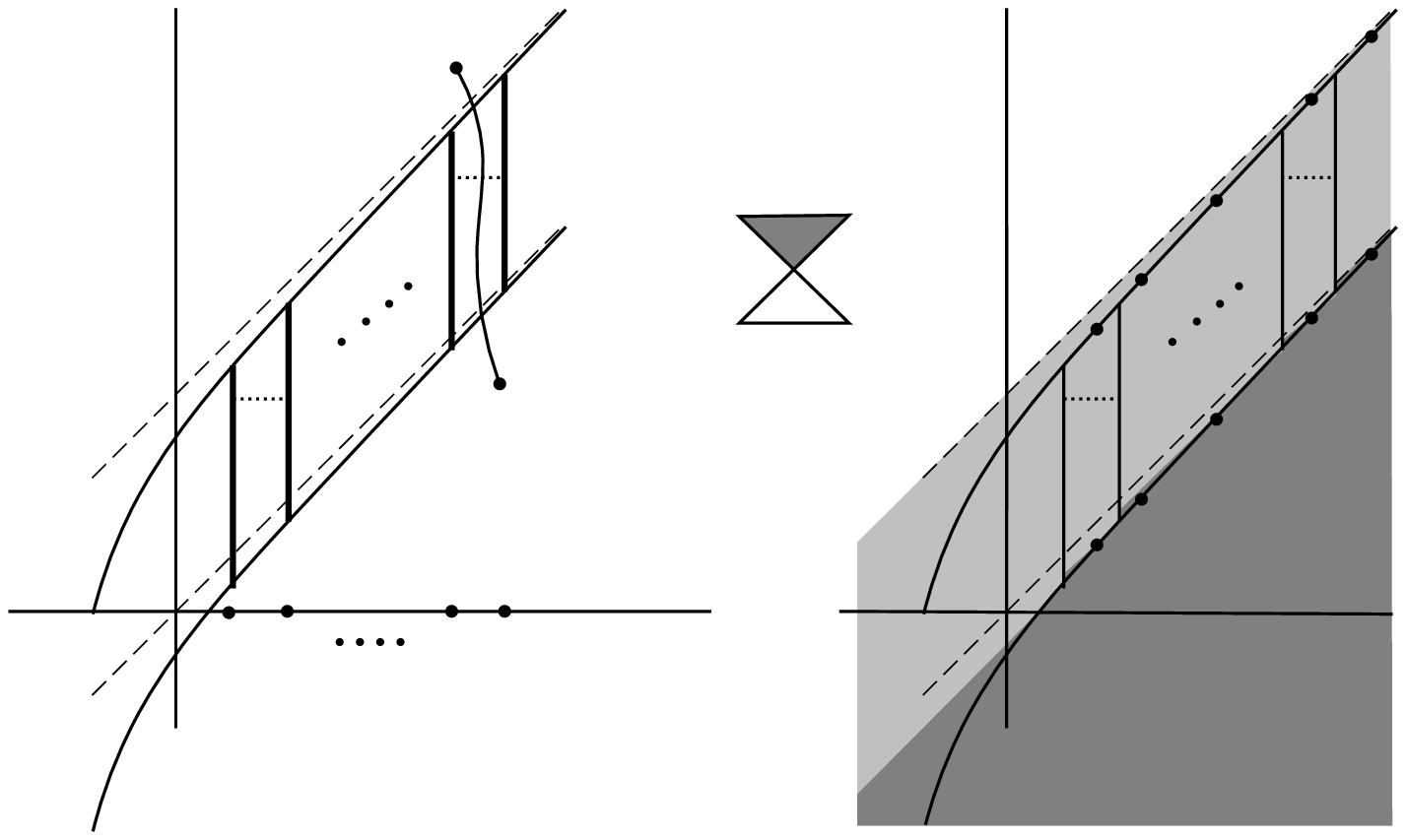}}
	\put(66.09,86.81){\fontsize{8.03}{9.64}\selectfont $1$}
	\put(82.77,86.81){\fontsize{8.03}{9.64}\selectfont $2$}
	\put(132.02,86.41){\fontsize{8.03}{9.64}\selectfont $n$}
	\put(146.32,86.81){\fontsize{8.03}{9.64}\selectfont $n+1$}
	\put(60.16,133.67){\fontsize{8.03}{9.64}\selectfont $S_1$}
	\put(88.29,146.78){\fontsize{8.03}{9.64}\selectfont $S_2$}
	\put(150.68,212.60){\fontsize{8.03}{9.64}\selectfont $S_{n+1}$}
	\put(150.99,226.41){\fontsize{8.03}{9.64}\selectfont $H_n$}
	\put(25.94,125.10){\fontsize{8.03}{9.64}\selectfont $\sigma_+$}
	\put(26.53,61.13){\fontsize{8.03}{9.64}\selectfont $\sigma_-$}
	\put(72.66,162.45){\fontsize{8.03}{9.64}\selectfont $H_1$}
	\put(150.73,160.64){\fontsize{6.42}{7.71}\selectfont $x$}
	\put(137.95,259.21){\fontsize{6.42}{7.71}\selectfont $y$}
	\put(145.64,178.09){\fontsize{8.03}{9.64}\selectfont $\gamma$}
	\put(322.02,109.99){\fontsize{8.03}{9.64}\selectfont $x_1$}
	\put(339.50,125.99){\fontsize{8.03}{9.64}\selectfont $x_2$}
	\put(389.25,178.11){\fontsize{8.03}{9.64}\selectfont $x_n$}
	\put(409.84,200.69){\fontsize{8.03}{9.64}\selectfont $x_{n+1}$}
	\put(310.48,187.20){\fontsize{8.03}{9.64}\selectfont $y_1$}
	\put(323.73,201.88){\fontsize{8.03}{9.64}\selectfont $y_2$}
	\put(373.40,253.58){\fontsize{8.03}{9.64}\selectfont $y_n$}
	\put(380.60,269.87){\fontsize{8.03}{9.64}\selectfont $y_{n+1}$}
	\put(95.72,9.13){\fontsize{16.06}{19.27}\selectfont (A)}
	\put(331.84,9.72){\fontsize{16.06}{19.27}\selectfont (B)}
	\put(123.50,186.73){\fontsize{8.03}{9.64}\selectfont $S_n$}
	\end{picture}%
	\fi		
		\caption{\label{fig:exHarris} The space $M$ is constructed in the following way: Consider in $\mathbb{L}^2$ two timelike curves $\sigma_-$ and $\sigma_+$ approaching two parallel lightlike lines, as we can see in (A). $M$ is obtained by removing from $\mathbb{L}^2$ the segments $S_n$ obtained by joining vertically the points $\sigma_-(n)$ and $\sigma_+(n)$. The universal cover $V$ of $M$ contains then a numerable family of copies of $M$ glued along the segments $H_i$ coherently (see details on Example \ref{ex:exemHarris} and \cite{H}).}
		
	\end{figure}

	The aim of this first example is twofold: On the one hand, we will present the main properties about the universal cover of a spacetime $M$ where we have removed a numerable family of compact segments. This behaviour will be used constantly on the forthcoming examples. On the other hand, it is a first example showing that $\hat{\pi}$ (and so, $\hat{\jmath}$) could be non-continuous in general. 
	
	Let us consider an spacetime $M$ as in the Figure \ref{fig:exHarris}, and let $V$ denote its universal cover. As it is described in the last example of \cite{H}, $V$ contains a numerable family of copies of $M$, that we will denote by $\{n\}\times M$ with $n\in \mathbb{Z}$, glued coherently along the segments $H_n$. For a given element $x\in M$, let us denote by $p$ its lift in $V$ living in the fibre $\{0\}\times M$. We will also denote by $n\cdot p$ the lift of $x$ in the fibre $\{n\}\times M$ (i.e., $p\equiv 0\cdot p$).
	
	In order to understand how the fibres are glued along $H_n$, let us show how the lifts of curves behave. Consider $\gamma$ a curve on $M$ as it is showed in Fig. \ref{fig:exHarris} (A), which is a timelike curve joining two points $x$ and $y$. Let $p$ and $q$ be the corresponding lifts in the fibre $\{0\}\times M$ and consider $\lev{\gamma}{}$ a lift of $\gamma$ on $V$ with start point $m\cdot p$. The fibres are glued in such a way that, as $\gamma$ intersects the segment $H_n$, the lifted curve $\lev{\gamma}{}$ moves from the fibre $\{m\}\times M$ to $\{m+n\}\times M$, being $(m+n)\cdot q$ its final point.
	
	\smallskip
	
	Once we have pointed out this behaviour, let us observe the particularities of the example regarding the continuity of $\hat{\pi}$. Let us observe now Fig. \ref{fig:exHarris} (B), where we have two TIPs $P\subsetneq P'$ defined by the sequences $\{x_n\}_{n}$ and $\{y_n\}_{n}$ ($P$ is filled in dark grey, while $P'$ has a lighter grey). It is not difficult to observe, due to the behaviour described before, that $m\cdot p_n\not\ll m\cdot q_n$ for any $m\in \mathbb{N}$. In fact, it follows that	
	\[
	m\cdot p_n\ll (m+n)\cdot q_n\quad\hbox{ for all $n\in\mathbb{N}$ }
	\]
as we can consider curve on $M$ joining $x_n$ with $y_n$ and intersecting $H_n$. From this, we can prove that: (a) the sequence $\lev{\sigma}{}=\{I^-(q_n)\}_{n}$ has $\lev{P'}{}$ (the lift of $P'$ on the fibre $\{0\}\times M$) on its limit, (b) the sequence $\{I^-(n\cdot q_n)\}_{n}$ has $\lev{P}{}$ on its limit (the inclusion on the inferior limit is straightforward, while the proof of the maximal character is detailed in \cite{H}) and (c) $\hat{\pi}(\lev{P}{})=P\subsetneq P'=\hat{\pi}(\lev{P'}{})$. In conclusion, and recalling Prop. \ref{thm:hatjcont}, $\hat{\pi}$ is not continuous.
	 
\end{exe}

\begin{exe}\label{ex:exe1}
	
	\begin{figure}
		\centering
		\ifpdf
		\setlength{\unitlength}{1bp}%
		\begin{picture}(275.17, 244.58)(0,0)
		\put(0,0){\includegraphics{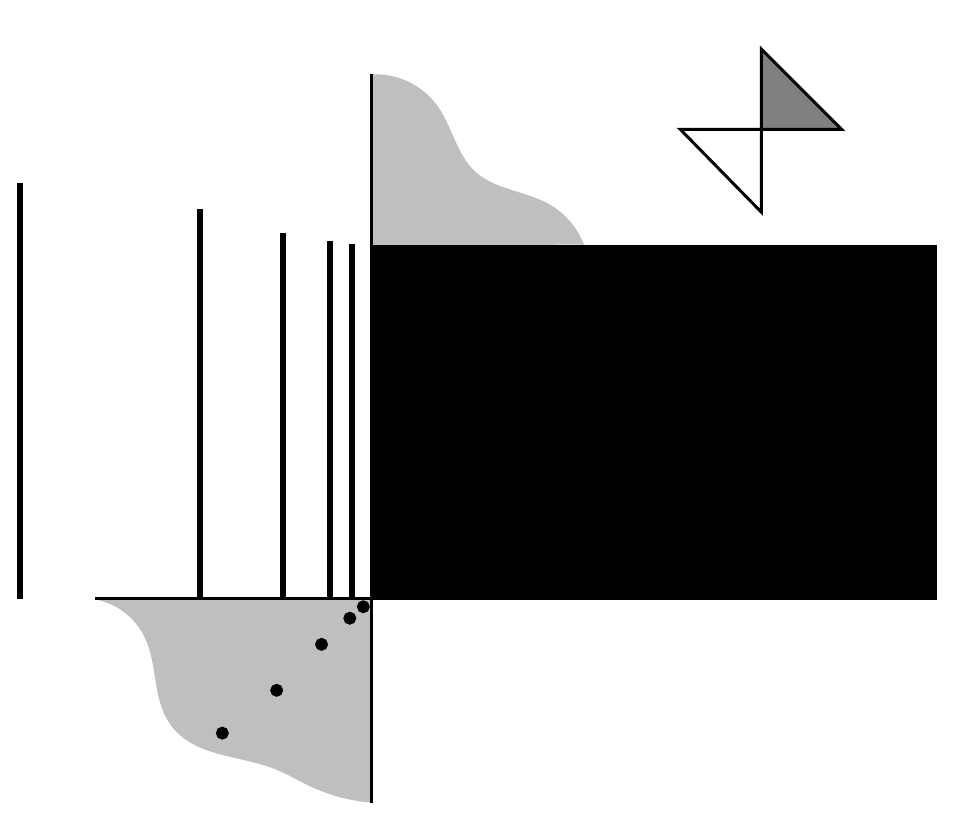}}
		\put(61.93,46.01){\fontsize{14.23}{17.07}\selectfont $P$}
		\put(115.17,175.07){\fontsize{14.23}{17.07}\selectfont $F$}
		\put(65.84,27.06){\fontsize{8.54}{10.24}\selectfont $x_n$}
		\put(60.07,175.03){\fontsize{11.38}{13.66}\selectfont $S_n$}
		\put(9.41,217.57){\fontsize{17.07}{20.49}\selectfont $M$}
		\end{picture}%
		\else
		\setlength{\unitlength}{1bp}%
		\begin{picture}(275.17, 244.58)(0,0)
		\put(0,0){\includegraphics{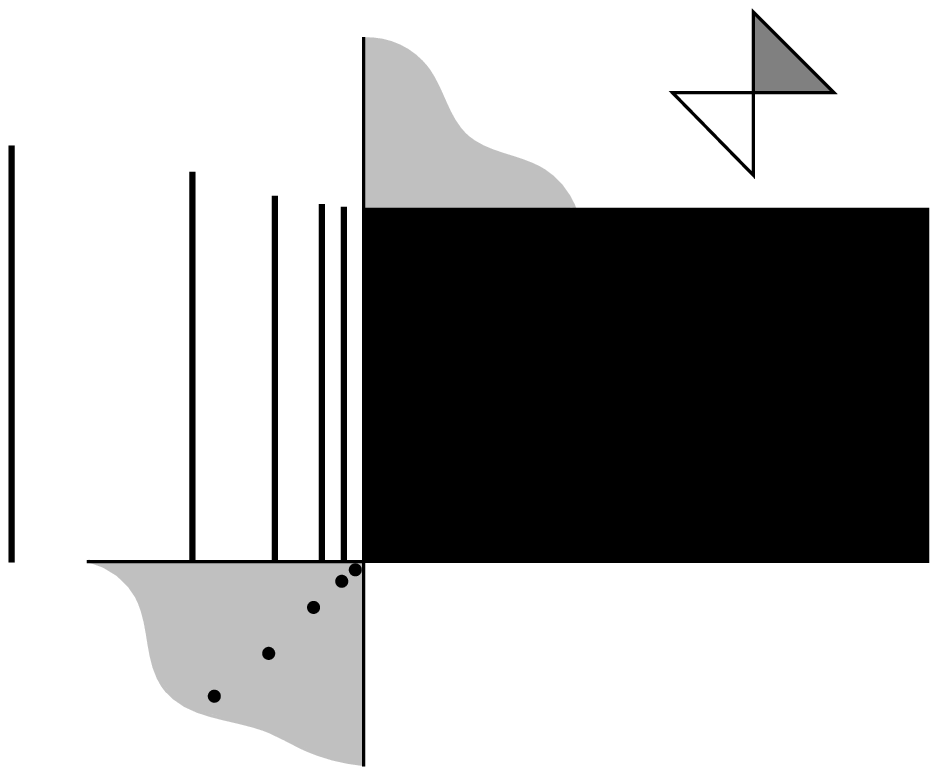}}
		\put(61.93,46.01){\fontsize{14.23}{17.07}\selectfont $P$}
		\put(115.17,175.07){\fontsize{14.23}{17.07}\selectfont $F$}
		\put(65.84,27.06){\fontsize{8.54}{10.24}\selectfont $x_n$}
		\put(60.07,175.03){\fontsize{11.38}{13.66}\selectfont $S_n$}
		\put(9.41,217.57){\fontsize{17.07}{20.49}\selectfont $M$}
		\end{picture}%
		\fi
		
		\caption{\label{fig:ex1} $M$ is constructed by removing from $\mathbb{L}^2$ the black square and the vertical segments $S_n$. As it was pointed out in \cite[Figure 11]{FHSFinalDef}, the terminal sets $P$ and $F$ are $S$-related, and so, they form a pair $(P,F)\in \overline{M}$. However, if $\lev{P}{}$ is a lift of $P$ to the universal cover $V$, it follows that $\uparrow \lev{P}{}=\emptyset$.}
		
	\end{figure}
	
Let us show now the optimality of Prop. \ref{prop:welldefpinotame} by showing that, even when the maps $\hat{\pi}$ and $\check{\pi}$ are continuous, it could happen that a point $(\lev{P}{},\emptyset)\in \overline{V}$ (resp. $(\emptyset,\lev{F}{})$) is not well projected (recall Prop. \ref{prop:proypares}). For this, in this example we will show a point $(P,F)\in M$ with no natural lift on $V$.  
	
	Let us consider $M$ a spacetime as described in Figure \ref{fig:ex1} and $V$ its universal cover. As it is pointed out in \cite[Figure 11]{FHSFinalDef}, both sets $P\sim_S F$ are $S$-related. Now, let us fix $\lev{P}{}$ and $\lev{F}{}$ lifts of the corresponding terminal sets on $\{0\}\times M\subset V$ as we have done on Example \ref{ex:exemHarris}; and denote by $\{p_n\}_{n} \subset \{0\}\times M$ a future chronological sequence which is lift of the sequence $\{x_n\}_{n}$ showed in Fig. \ref{fig:ex1}. Recall that the lifts on $V$ of timelike curves of $M$ moving between $S_n$ and $S_{n+1}$ behave essentially as described in Example \ref{ex:exemHarris}. Hence, it follows that 
	
	\[
	\cap_{n\in\mathbb{N}} I^+(p_n)=\emptyset. 
	\] 
	Therefore, the set $\uparrow \lev{P}{}$ is empty and $\lev{P}{}\sim_S \emptyset$. 
	
	However, it is not difficult to see that both $\hat{\pi}$ and $\check{\pi}$ are continuous. Recall that the non-continuity of such maps can only follow by the existence of a sequence $\{y_n\}_{n} \subset M$ admitting divergent lifts. 
	
	The only case we have to be concerned is when $\{y_n\}_{n}$ converges on $\mathbb{R}^2$, or to the point $(0,1)$ or to $(0,0)$ (in the other cases, the convergence is essentially the usual one in $\mathbb{R}^2$). Assume for instance that the sequence $\{y_n\}_{n}$ converges to the point $(0,1)$ (the other case is completely analogous). It is straightforward to check that any convergent lift with the past chronological topology of $\{y_n\}_{n}$ in $V$ are, up to a subsequence, of the form $\{m\cdot q_n\}_n$, with $m\in \mathbb{Z}$ constant and $q_n \in \{0\} \times M\subset V$ a fixed lift of $\{y_n\}_{n}$. In particular, their limits are of the form $m\cdot \lev{F}{}$. This is due the fact that the IFs involved will not have points between the segments $S_n$, and so, we do not have to move between different fibres of $V$. Therefore, any convergent lift of $\{y_n\}_{n}$ with the past topology converges to a terminal set on $\check{\pi}^{-1}(F)$, and so, $\{y_n\}_{n}$ does not have past divergent lifts (condition (ii) in Definition \ref{def:divlif} cannot be fulfilled). 
	
	For the future topology however the situation is a little more technical, as the involved IPs contain these points between segments $S_n$. With some effort, it can be proved that if ${\rm LI}(\{I^-(g_n\,q_n)\}_{n}) \neq \emptyset$ for some $\{g_n\}_{n} \subset \mathbb{Z}$, then ${\rm LI}(\{I^-(g_n\,q_n)\}_{n})=m\,\lev{P}{}$ for some $m\in \mathbb{Z}$. In particular, any convergent lift with the future topology of $\{y_n\}_{n}$ will converge to some TIP on $\hat{\pi}^{-1}(P)$, and so, reasoning as in previous case, $\{y_n\}_{n}$ does not admits future divergent lifts.
	
	In conclusion, $M$ does not admit (future or past) divergent lifts, and so, both $\hat{\pi}$ and $\check{\pi}$ are continuous. 
\end{exe} 

\begin{exe}\label{ex:exe2}
	
\begin{figure}
	\centering
	\ifpdf
	\setlength{\unitlength}{1bp}%
	\begin{picture}(431.46, 201.36)(0,0)
	\put(0,0){\includegraphics{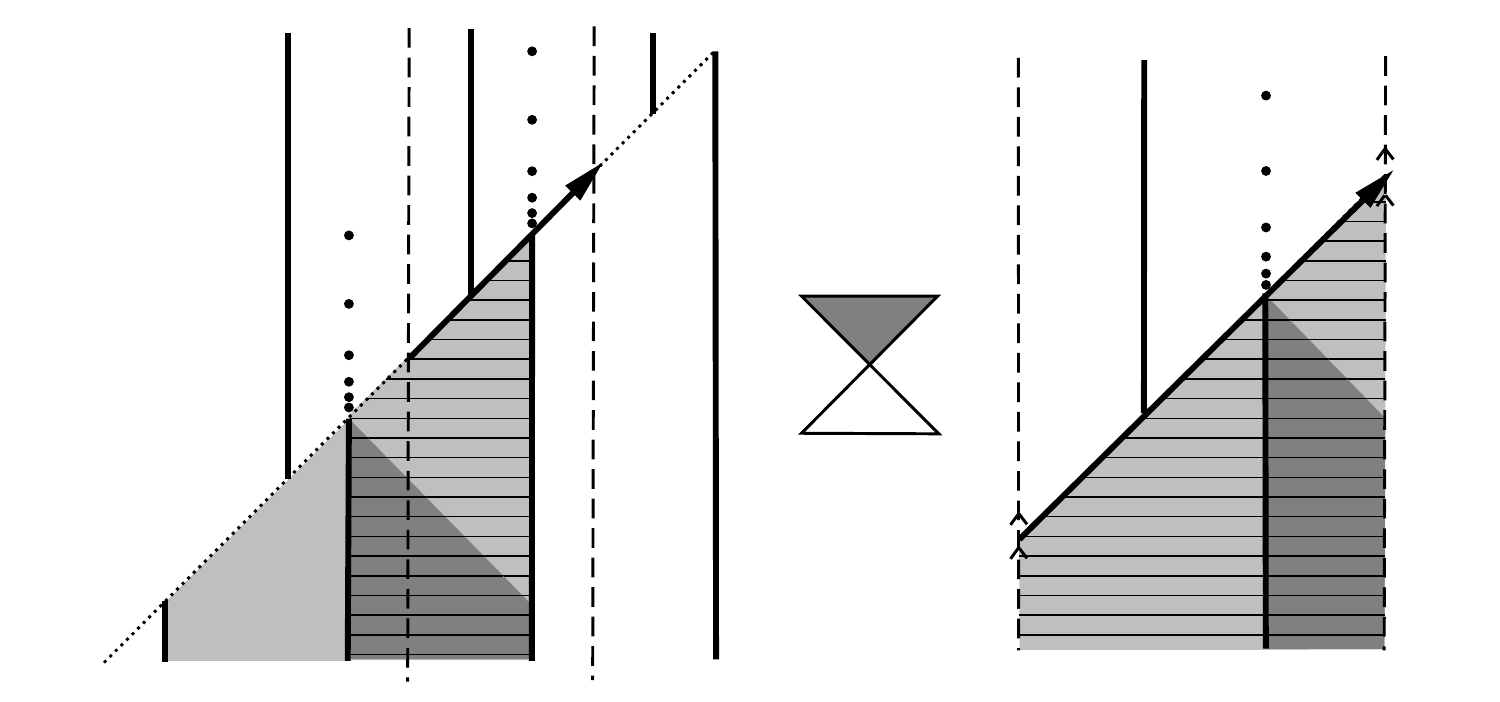}}
	\put(99.42,117.47){\fontsize{8.83}{10.60}\selectfont $p_n$}
	\put(151.47,170.48){\fontsize{8.83}{10.60}\selectfont $1\,p_n$}
	\put(70.45,35.66){\fontsize{8.83}{10.60}\selectfont $\lev{P}{1}$}
	\put(104.98,26.91){\fontsize{8.83}{10.60}\selectfont $\lev{P}{2}$}
	\put(128.16,98.58){\fontsize{8.83}{10.60}\selectfont $1\,\lev{P}{1}$}
	\put(324.48,58.42){\fontsize{8.83}{10.60}\selectfont $P_1$}
	\put(371.12,62.84){\fontsize{8.83}{10.60}\selectfont $P_2$}
	\put(365.43,154.07){\fontsize{8.83}{10.60}\selectfont $x_n$}
	\put(137.85,190.58){\fontsize{6.54}{7.85}\selectfont $r^1_0$}
	\put(154.73,15.40){\fontsize{6.54}{7.85}\selectfont $r^2_0$}
	\put(268.81,40.01){\fontsize{8.54}{10.24}\selectfont $(0,0)$}
	\put(400.99,149.17){\fontsize{8.54}{10.24}\selectfont $(1,1)$}
	\end{picture}%
	\else
	\setlength{\unitlength}{1bp}%
	\begin{picture}(431.46, 201.36)(0,0)
	\put(0,0){\includegraphics{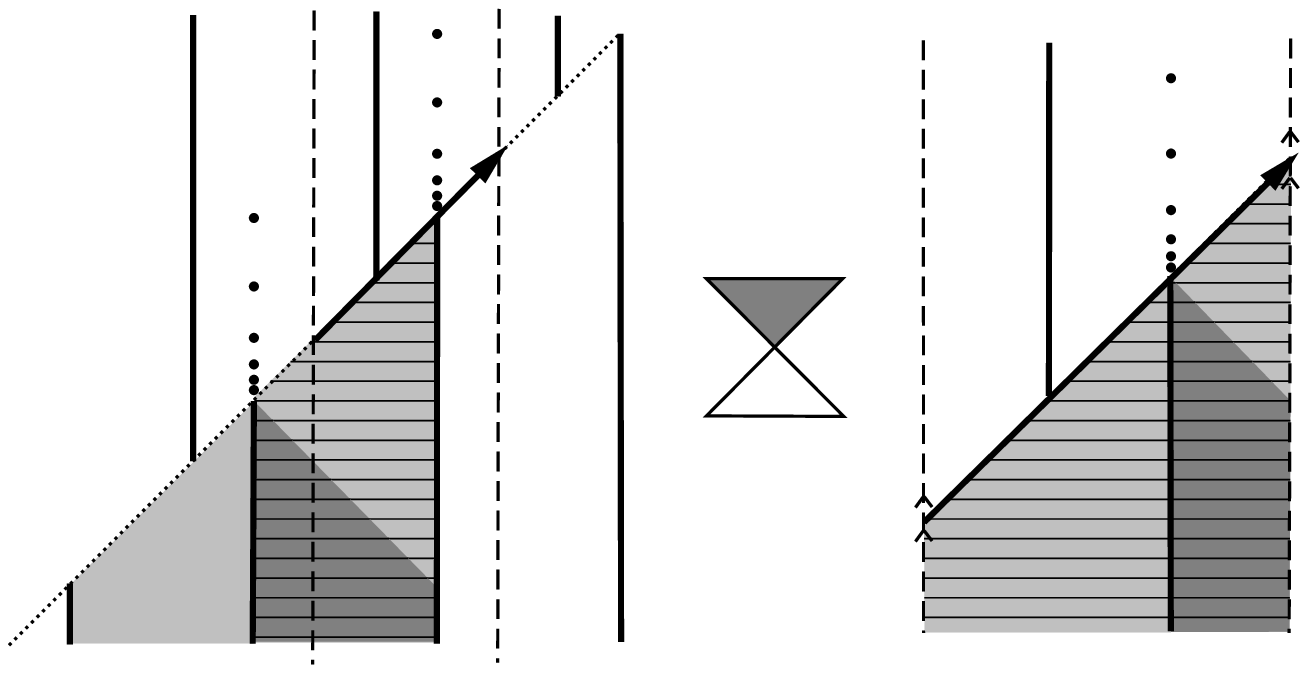}}
	\put(99.42,117.47){\fontsize{8.83}{10.60}\selectfont $p_n$}
	\put(151.47,170.48){\fontsize{8.83}{10.60}\selectfont $1\,p_n$}
	\put(70.45,35.66){\fontsize{8.83}{10.60}\selectfont $\lev{P}{1}$}
	\put(104.98,26.91){\fontsize{8.83}{10.60}\selectfont $\lev{P}{2}$}
	\put(128.16,98.58){\fontsize{8.83}{10.60}\selectfont $1\,\lev{P}{1}$}
	\put(324.48,58.42){\fontsize{8.83}{10.60}\selectfont $P_1$}
	\put(371.12,62.84){\fontsize{8.83}{10.60}\selectfont $P_2$}
	\put(365.43,154.07){\fontsize{8.83}{10.60}\selectfont $x_n$}
	\put(137.85,190.58){\fontsize{6.54}{7.85}\selectfont $r^1_0$}
	\put(154.73,15.40){\fontsize{6.54}{7.85}\selectfont $r^2_0$}
	\put(268.81,40.01){\fontsize{8.54}{10.24}\selectfont $(0,0)$}
	\put(400.99,149.17){\fontsize{8.54}{10.24}\selectfont $(1,1)$}
	\end{picture}%
	\fi
	\caption{\label{fig:ex2} The space $V$ (on the left) is $\mathbb{L}^2$ with two families of lines $\{r_n^1\}_n$ and $\{r_n^2\}_n$ removed, where $r_n^1=\{(1/3+n,y):y\geq 1/3+n\}$ and $r_n^2=\{(2/3+n,y):y\leq 2/3+n\}$. The action of an element $g$ of the group $\mathbb{Z}$ is just a translation of $g$-times the vector $(1,1)$. The quotient space $M=V/\mathbb{Z}$ (on the right) is the space $(0,1)\times \mathbb{R}$ with the points $(0,y)$ and $(1,y+1)$ identified.}
\end{figure}
	The optimality of several previous results is stressed here. For instance, it shows that the non-continuity of $\hat{\pi}$ can be obtained by considering a sequence $\{g_n\}_{n}$ of elements of the group constant. It also shows that if a past chronological sequence has future divergent lifts, then the thesis of Prop. \ref{prop:proypares} can fail. Finally, it shows that the finitely chronology property is not enough to ensure the continuity of the partial maps $\hat{\pi}$ and $\check{\pi}$, being necessary to include them additionally.
	
	\smallskip 
	
	Let us consider an space $V\subset \mathbb{R}^2$ as showed in Figure \ref{fig:ex2}. On such a space, consider $G\equiv \mathbb{Z}$ an isometry group given by the following action:
		
	\[\begin{array}{rcc}
	\mathbb{Z}\times V& \rightarrow& V\\
	(g,p)&\rightarrow& g\cdot p:=p+g(1,1).
\end{array}
	\]
        
	\noindent The quotient $M=V/\mathbb{Z}$ can be seen as a cylinder with some cuts on it (see Figure \ref{fig:ex2} (B)). Let us summarize the properties of the spacetime covering projection $\pi:V\rightarrow M$. On the one hand, and observing Figure \ref{fig:ex2} (B), it follows easily that $\overline{M}$ contains the pairs $(P_1,F)$ and $(P_2,\emptyset)$. Indeed, both sets $P_1,P_2$ are contained in $\downarrow F$, but thanks to the identification of both lateral sides, it follows that $P_2\subset P_1$, so only $P_1$ is maximal on the common past of $F$. However, on $\overline{V}$ we have both pairs $(\lev{P}{1},\lev{F}{})$ and $(\lev{P}{2},\lev{F}{})$, so the thesis on Prop. \ref{prop:proypares} is false on this case.
	
	On the other hand, we can see that the past chronological sequence $\{x_n\}_{n}$ depicted on Figure \ref{fig:ex2} (B) has future divergent lifts. In fact, consider the sequences $\{p_n\}_{n}$ and $\{1\,p_n\}_{n}$ in Figure \ref{fig:ex2} (A), which are both lifts of $\{x_n\}_{n}$. The first sequence converges with the future topology to $\lev{P}{2}$, while the second one converges to $1\,\lev{P_1}{}$. Moreover, it follows (as we have reasoned before) that $\hat{\pi}(\lev{P}{2})\subsetneq \hat{\pi}(\lev{P}{1})$. Therefore, the sequences $\{p_n\}_{n}, \{1\,p_n\}_{n}$ and the points $\lev{P}{1},\lev{P}{2}\in \hat{M}$ fulfil the conditions on Definition \ref{def:divlif} and $\{x_n\}_{n}$ admits future divergent lifts. In particular, we deduce that $\hat{\pi}$ is not continuous.
	
	Finally, it is quite straightforward to see that $(V,G)$ is finitely chronological. Observe that, if $p\ll q$ in $V$, it could exists (at most) one element in $g\in \mathbb{Z}$ such that $p\ll g\,q$ (specifically, $g=\pm 1$).
		
%
%
\end{exe}

\begin{exe}\label{ex:tame} 
	
	\begin{figure}
		\centering
		\ifpdf
		\setlength{\unitlength}{1bp}%
		\begin{picture}(484.81, 203.99)(0,0)
		\put(0,0){\includegraphics{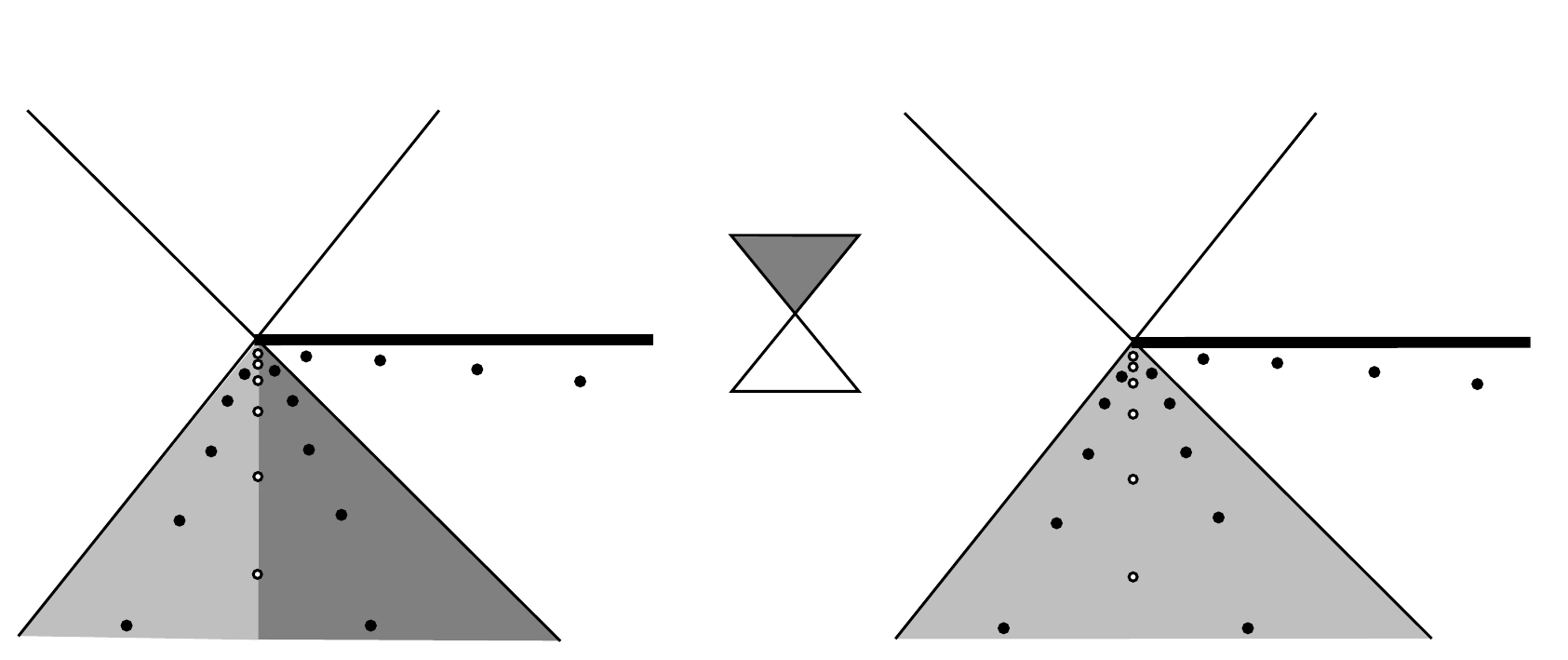}}
		\put(340.37,157.73){\fontsize{14.23}{17.07}\selectfont $F$}
		\put(311.15,12.56){\fontsize{8.54}{10.24}\selectfont $x'_1$}
		\put(327.53,45.56){\fontsize{8.54}{10.24}\selectfont $x'_2$}
		\put(337.34,66.58){\fontsize{8.54}{10.24}\selectfont $x'_3$}
		\put(386.85,12.06){\fontsize{8.54}{10.24}\selectfont $x_1$}
		\put(377.52,46.18){\fontsize{8.54}{10.24}\selectfont $x_2$}
		\put(367.94,67.08){\fontsize{8.54}{10.24}\selectfont $x_3$}
		\put(311.09,23.77){\fontsize{14.23}{17.07}\selectfont $P$}
		\put(68.91,158.54){\fontsize{14.23}{17.07}\selectfont $\lev{F}{}$}
		\put(40.07,12.61){\fontsize{8.54}{10.24}\selectfont $p'_1$}
		\put(56.46,45.60){\fontsize{8.54}{10.24}\selectfont $p'_2$}
		\put(65.88,67.01){\fontsize{8.54}{10.24}\selectfont $p'_3$}
		\put(116.16,12.11){\fontsize{8.54}{10.24}\selectfont $p_1$}
		\put(106.06,46.61){\fontsize{8.54}{10.24}\selectfont $p_2$}
		\put(96.48,67.90){\fontsize{8.54}{10.24}\selectfont $p_3$}
		\put(39.60,24.58){\fontsize{14.23}{17.07}\selectfont $\lev{P'}{}$}
		\put(104.97,25.93){\fontsize{14.23}{17.07}\selectfont $\lev{P}{}$}
		\put(457.87,89.20){\fontsize{8.54}{10.24}\selectfont $y_1$}
		\put(427.85,89.82){\fontsize{8.54}{10.24}\selectfont $y_2$}
		\put(398.74,89.82){\fontsize{8.54}{10.24}\selectfont $y_3$}
		\put(180.20,90.01){\fontsize{8.54}{10.24}\selectfont $q_1$}
		\put(150.18,90.63){\fontsize{8.54}{10.24}\selectfont $q_2$}
		\put(121.07,90.63){\fontsize{8.54}{10.24}\selectfont $q_3$}
		\put(300.15,184.49){\fontsize{17.07}{20.49}\selectfont $M$}
		\put(21.01,184.99){\fontsize{17.07}{20.49}\selectfont $V$}
		\end{picture}%
		\else
		\setlength{\unitlength}{1bp}%
		\begin{picture}(484.81, 203.99)(0,0)
		\put(0,0){\includegraphics{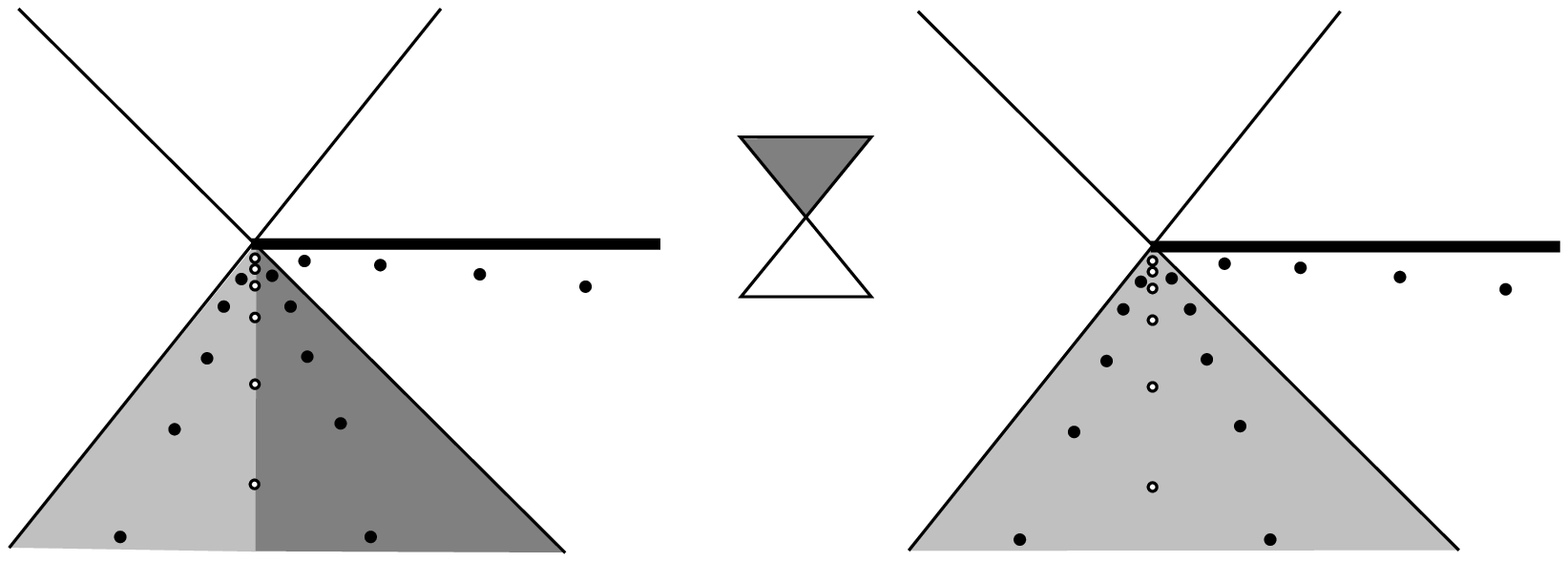}}
		\put(340.37,157.73){\fontsize{14.23}{17.07}\selectfont $F$}
		\put(311.15,12.56){\fontsize{8.54}{10.24}\selectfont $x'_1$}
		\put(327.53,45.56){\fontsize{8.54}{10.24}\selectfont $x'_2$}
		\put(337.34,66.58){\fontsize{8.54}{10.24}\selectfont $x'_3$}
		\put(386.85,12.06){\fontsize{8.54}{10.24}\selectfont $x_1$}
		\put(377.52,46.18){\fontsize{8.54}{10.24}\selectfont $x_2$}
		\put(367.94,67.08){\fontsize{8.54}{10.24}\selectfont $x_3$}
		\put(311.09,23.77){\fontsize{14.23}{17.07}\selectfont $P$}
		\put(68.91,158.54){\fontsize{14.23}{17.07}\selectfont $\lev{F}{}$}
		\put(40.07,12.61){\fontsize{8.54}{10.24}\selectfont $p'_1$}
		\put(56.46,45.60){\fontsize{8.54}{10.24}\selectfont $p'_2$}
		\put(65.88,67.01){\fontsize{8.54}{10.24}\selectfont $p'_3$}
		\put(116.16,12.11){\fontsize{8.54}{10.24}\selectfont $p_1$}
		\put(106.06,46.61){\fontsize{8.54}{10.24}\selectfont $p_2$}
		\put(96.48,67.90){\fontsize{8.54}{10.24}\selectfont $p_3$}
		\put(39.60,24.58){\fontsize{14.23}{17.07}\selectfont $\lev{P'}{}$}
		\put(104.97,25.93){\fontsize{14.23}{17.07}\selectfont $\lev{P}{}$}
		\put(457.87,89.20){\fontsize{8.54}{10.24}\selectfont $y_1$}
		\put(427.85,89.82){\fontsize{8.54}{10.24}\selectfont $y_2$}
		\put(398.74,89.82){\fontsize{8.54}{10.24}\selectfont $y_3$}
		\put(180.20,90.01){\fontsize{8.54}{10.24}\selectfont $q_1$}
		\put(150.18,90.63){\fontsize{8.54}{10.24}\selectfont $q_2$}
		\put(121.07,90.63){\fontsize{8.54}{10.24}\selectfont $q_3$}
		\put(300.15,184.49){\fontsize{17.07}{20.49}\selectfont $M$}
		\put(21.01,184.99){\fontsize{17.07}{20.49}\selectfont $V$}
		\end{picture}%
	 	\fi
	 	\caption{\label{fig:ex4} The space $M$ (on the right) is $\mathbb{L}^2$ with the line $r=\{(x,0):x \geq 0\}$ and the sequence of points $\{(0,-\frac{1}{n})\}_n$ removed, while $V$ is the universal cover of $M$. On the first one, associated to the point $(0,0)$ we have the point $(P,F)\in \partial M$. However, the set $P$ lifts to a fixed fibre $\{0\}\times M \subset V$ as two different terminal past sets $\lev{P}{}$ and $\lev{P'}{}$, creating two different points $(\lev{P'}{},F),(\lev{P}{},\emptyset)\in \partial V$. In particular, it follows that the sequence $\{y_n\}_n$ depicted on the right is not convergent, while its lifts $\{q_n\}_n$ converges to $(\lev{P}{},\emptyset)$.}
	 \end{figure}
	
The aim of this example is threefold. On the one hand, it will give an example of a non tame spacetime covering projection. On the other, it will show that even when $M$ does not admit constant sequences with future divergent lifts, the $G$-orbits can be non closed.
 Finally, it will also show that the proof of Prop. \ref{prop:contparimtot} is optimal with respect to the condition on the projection. In fact, we will show the existence of a sequence $\{q_n\}_{n} \subset V$ and a TIP $\lev{P}{}\in \hat{V}$ with $\lev{P}{}\sim_{S}\emptyset$ and such that $\overline{P}{}\in \hat{L}_{V}(\{I^-(q_n)\}_{n})$, $P\in \hat{L}_{M}(\{I^-(y_n)\}_{n})$ but $P\sim_{S} F$ with $F\neq \emptyset$. 

\smallskip 

	Let us consider the Lorentz manifold \[M=\mathbb{L}^2\setminus\left(\{[0,\infty)\times 0\}\cup \cup_{n} \{(0,-\frac{1}{n})\}\right)\]
	(see Figure \ref{fig:ex4}), and take $V$ its universal cover. The behaviour of the lifts of curves in $M$ to $V$ behaves essentially in the same manner described in Example \ref{ex:exemHarris}, that is, it contains a numerable family of copies of $M$ (which will be denoted again by $\{n\}\times M$) glued together accordingly; and whenever a curve $\gamma\subset M$ pass between two holes of $M$, the initial point and the endpoint of the lifted curve $\lev{\gamma}{}$ live in two different fibres of such a numerable family.
	
	It follows that the point $(0,0)\in \mathbb{R}^2$ has associated in $\overline{M}$ a singular point $(P,F)\in \overline{M}$. However, the lift of the terminal set $P$ in a concrete fibre, say $\{0\}\times M$, determines two different terminal sets $\lev{P}{},\lev{P'}{}$. The reason is simple, any timelike curve joining a point of the sequence $\{x_i\}_{i}$ with $\{x'_i\}_{i}$ should pass between two holes of $M$, and so, its lift moves along different fibres. Moreover, from construction, we have that for each $p_i$ there exists $g_i$ ensuring that $p_i\in g_i\,\lev{P'}{}$. However, the sequence $\{g_i\}_{i}$ cannot be considered constant (not even up to a subsequence), so there is no $g\in G$ such that $\lev{P}{}\subset g\,\lev{P'}{}$ and the projection cannot be tame. Moreover, it follows from the construction that $\lev{P}{}\subset {\rm LI}(\{g_n\,\lev{P'}{}\}_n)$ and it is maximal on the superior limit, i.e., $\lev{P}{}\in \hat{L}_{V}(\{g_n\,\lev{P'}{}\}_n)$. Therefore, the $G$-orbit $\{g\,\lev{P'}{}\}_{g\in G}$ is not closed as $\lev{P}{}$ is an element not belonging to the $G$-orbit of $\lev{P'}{}$ but which is in its closure.
	
	Let us now show the existence of a sequence $\{q_n\}_{n}$ as described in the first paragraph of the example. Consider a sequence $\{y_n\}_{n}$ as in Figure \ref{fig:ex4} and $\{q_n\}_{n}$ its lift in the fibre $\{0\}\times M\subset V$. As we can see in the figure, $P\in \hat{L}_{M}(\{I^-(y_n)\}_{n})$ and $\lev{P}{}\in \hat{L}_{V}(\{I^-(q_n)\}_{n})$. Moreover, as we have mention before, $P\sim_S F$ with $F\neq\emptyset$. So, it only rest to show that $\lev{P}{}\sim_S\emptyset$. But this follows from the fact that $\uparrow \lev{P}{}=\emptyset$ (recall that whenever a timelike curve moves through the space between two holes, it pass to another fibre in $V$).  
%
%
	Summarizing, we have shown in particular that the map $\overline{\pi}$ is not continuous. The sequence $\{q_n\}_{n}$ converges to the point $(\lev{P}{},\emptyset)\in \overline{V}$, while its projection $\{y_n\}_{n}$ does not converge to $(P,F)\in \overline{M}$ (note that ${\rm LI}(\{I^+(y_n)\}_{n})=\emptyset$).
	
%
\end{exe} 
 
 \begin{exe}\label{ex:exe3}
 	
 	\begin{figure}
 		\centering
 		\ifpdf
 		\setlength{\unitlength}{1bp}%
 		\begin{picture}(231.61, 216.61)(0,0)
 		\put(0,0){\includegraphics{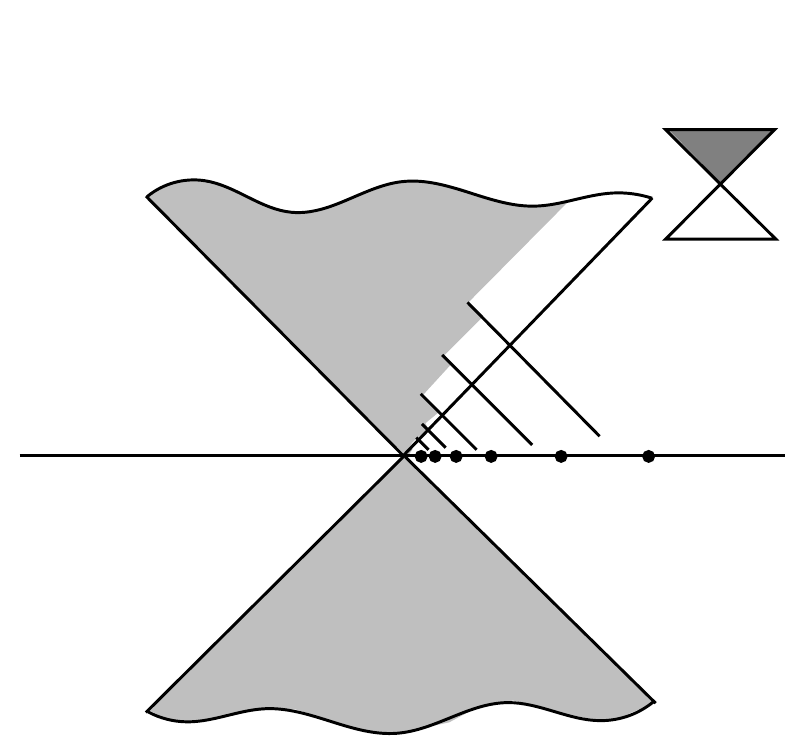}}
 		\put(106.42,19.53){\fontsize{14.23}{17.07}\selectfont $P$}
 		\put(102.64,138.73){\fontsize{14.23}{17.07}\selectfont $F$}
 		\put(184.04,76.25){\fontsize{11.69}{12.83}\selectfont $x_n$}
 		\put(160.35,104.71){\fontsize{14.54}{15.24}\selectfont $S_n$}
 		\put(92.81,193.16){\fontsize{18.76}{23.31}\selectfont $M$}
 		\end{picture}%
 		\else
 		\setlength{\unitlength}{1bp}%
 		\begin{picture}(231.61, 216.61)(0,0)
 		\put(0,0){\includegraphics{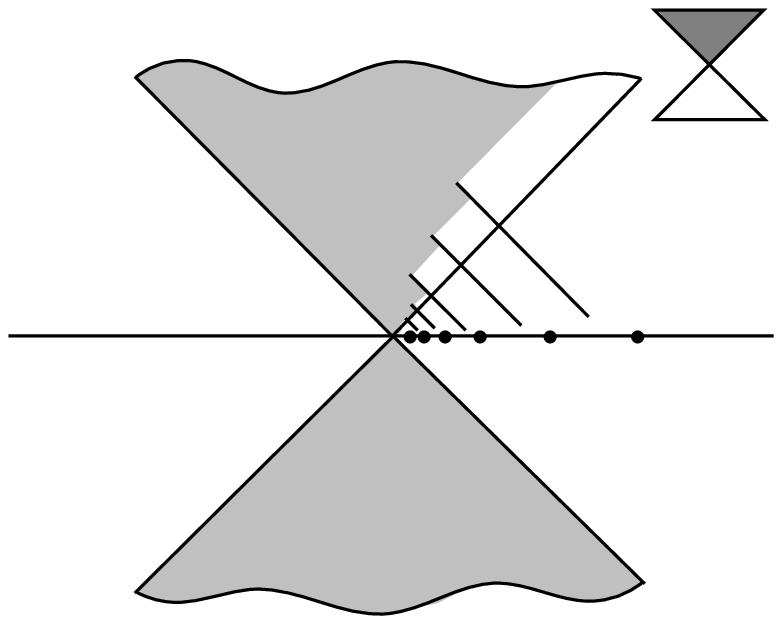}}
 		\put(106.42,19.53){\fontsize{14.23}{17.07}\selectfont $P$}
 		\put(102.64,138.73){\fontsize{14.23}{17.07}\selectfont $F$}
 		\put(184.04,79.25){\fontsize{5.69}{6.83}\selectfont $x_n$}
 		\put(160.35,104.71){\fontsize{8.54}{10.24}\selectfont $S_n$}
 		\put(92.81,193.16){\fontsize{22.76}{27.31}\selectfont $M$}
 		\end{picture}%
 		\fi
 		\caption{\label{fig:ex3} Let $M$ be $\mathbb{L}^2$ with the segments $S_n$ removed. On this example, the sets $P$ and $F$ are $S$-related and the sequence $\{x_n\}_{n}$ has $(P,F)$ on its limit. However, for $n$ big enough, the curves from $x_n$ to points in $F$ should pass between $S_n$ and $S_{n+1}$. In particular, if $\{p_{n}\}_{n} \subset \{0\}\times M$ is a lift of $\{x_n\}_{n}$, $g\,\lev{F}{}\not\subset {\rm LI}(\{I^+(p_n)\}_{n})$ for any $g\in G$. }
 		
 	\end{figure}
 	
 	This simple example will stress that the openness of $\hat{\jmath}$ and $\check{\jmath}$ is not enough to ensure the openness of $\overline{\jmath}$, even 
 	 when $\overline{\pi}$ is well defined and both $\hat{\jmath}$ and $\check{\jmath}$ are continuous. 
 	
 	Let us consider $M$ an spacetime as in Fig. \ref{fig:ex3} and $V$ the universal cover of $M$. On $M$, both sets $P$ and $F$ are $S$-related and the sequence $\{x_n\}_{n}$ converges to the point $(P,F)$. On $V$, and thanks that we can take curves joining points from $P$ to $F$ without moving between any $S_n$ and $S_{n+1}$, we can obtain lifts $\lev{P}{}$ and $\lev{F}{}$ with $\lev{P}{}\sim_S\lev{F}{}$ (we can assume that both sets live in the fibre $\{0\}\times M$). 
 	
 	However, no lift of the sequence $\{x_n\}_{n}$ converges to $(\lev{P}{},\lev{F}{})$. In fact, let us take $\{p_n\}_{n}$ a fixed lift of $\{x_n\}_{n}$ contained in $\{0\}\times M$. It is not difficult to observe that this lift is the only one satisfying that $\lev{P}{}\in \hat{L}_{V}(\{I^-(p_n)\}_{n})$. Even so, it is not true that $\lev{F}{}\in \check{L}_{V}(\{I^+(p_n)\}_{n})$, as any timelike curve joining a point $x_n$ with $F$ should pass through two lines $S_n$, hence its lift moves between two different fibres. Therefore, the sequence $\{x_n\}_{n}$ has no natural convergent lift and the map $\overline{\jmath}$ is not open.
 	
 	 Finally, let us observe that $\hat{\jmath}$ and $\check{\jmath}$ is continuous. This follows by reasoning as in Example \ref{ex:exe1}, recalling that the only cases where the continuity could fail is considering sequences $\{y_n\}_{n}$ converging to $(0,0)$.
 	

 \end{exe}

\begin{exe}
	\label{ex:ex5}
	\begin{figure}
		\centering
		\ifpdf
		\setlength{\unitlength}{1bp}%
		\begin{picture}(493.37, 215.16)(0,0)
		\put(0,0){\includegraphics{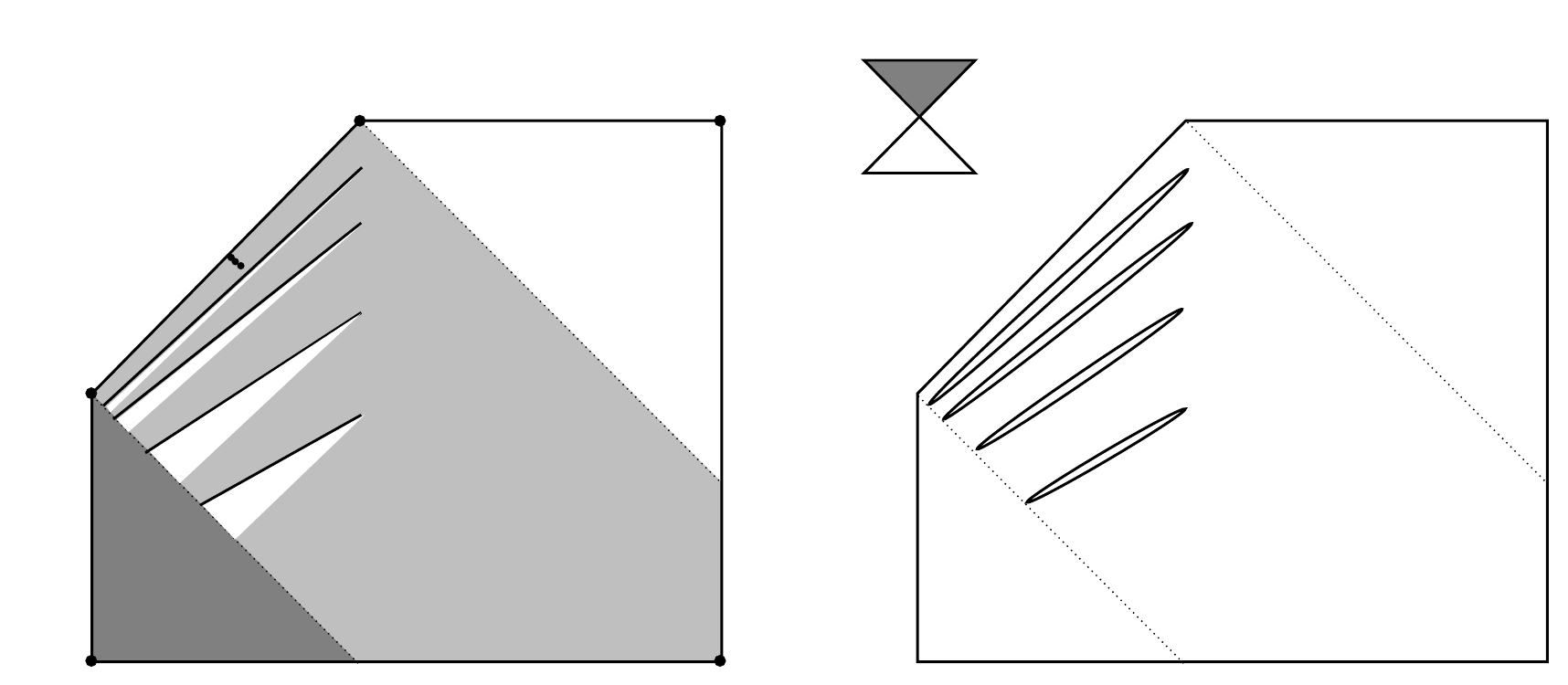}}
		\put(166.40,198.38){\rotatebox{360.00}{\fontsize{14.23}{17.07}\selectfont \smash{\makebox[0pt][l]{$M$}}}}
		\put(155.17,62.80){\rotatebox{360.00}{\fontsize{14.23}{17.07}\selectfont \smash{\makebox[0pt][l]{$P'_0$}}}}
		\put(37.25,32.26){\rotatebox{360.00}{\fontsize{14.23}{17.07}\selectfont \smash{\makebox[0pt][l]{$P$}}}}
		\put(4.71,90.57){\fontsize{8.54}{10.24}\selectfont $(0,0)$}
		\put(106.45,179.89){\fontsize{8.54}{10.24}\selectfont $(1,1)$}
		\put(106.45,129.89){\fontsize{12.54}{14.24}\selectfont $S_n$}
		\put(224.67,179.89){\fontsize{8.54}{10.24}\selectfont $(2,1)$}
		\put(230.27,8.51){\fontsize{8.54}{10.24}\selectfont $(2,-1)$}
		\put(-3.67,8.55){\fontsize{8.54}{10.24}\selectfont $(0,-1)$}
		\put(379.40,198.38){\fontsize{14.23}{17.07}\selectfont $\partial M$}
		\end{picture}%
		\else
		\setlength{\unitlength}{1bp}%
		\begin{picture}(493.37, 215.16)(0,0)
		\put(0,0){\includegraphics{ex5.pdf}}
		\put(166.40,198.38){\rotatebox{360.00}{\fontsize{14.23}{17.07}\selectfont \smash{\makebox[0pt][l]{$M$}}}}
		\put(155.17,62.80){\rotatebox{360.00}{\fontsize{14.23}{17.07}\selectfont \smash{\makebox[0pt][l]{$P'_0$}}}}
		\put(37.25,32.26){\rotatebox{360.00}{\fontsize{14.23}{17.07}\selectfont \smash{\makebox[0pt][l]{$P$}}}}
		\put(4.71,90.57){\fontsize{8.54}{10.24}\selectfont $(0,0)$}
		\put(106.45,179.89){\fontsize{8.54}{10.24}\selectfont $(1,1)$}
		\put(106.45,129.89){\fontsize{12.54}{14.24}\selectfont $S_n$}
		\put(224.67,179.89){\fontsize{8.54}{10.24}\selectfont $(2,1)$}
		\put(230.27,8.51){\fontsize{8.54}{10.24}\selectfont $(2,-1)$}
		\put(-3.67,8.55){\fontsize{8.54}{10.24}\selectfont $(0,-1)$}
		\put(379.40,198.38){\fontsize{14.23}{17.07}\selectfont $\partial M$}
		\end{picture}%
		\fi
		\caption{\label{fig:fig4} Let $M$ be an open set of $\mathbb{L}^2$ with the segments $S_n$ removed as in the left. Even if the segments are spacelike, the terminal set $P'_0$ (the past of the boundary point $(1,1)$) contains $P_0$ (the past of $(0,0)$).\protect \\ 
		The c-boundary of $M$ is represented on the right of the figure. Observe that, in the c-boundary, each segment $S_n$ is represented by a thin ellipse. This is due the fact that any non-extremal point of the segment is reachable by a future and past inextensible timelike curve, but the corresponding terminal sets are not $S$-related. So, such points are represented in the c-boundary as two points of the form $(P,\emptyset)$ and $(\emptyset,F)$. Only on the extremal points the corresponding TIP and TIF are $S$-related, and so, they determine only one point in the c-boundary.}
	\end{figure}

	Finally, we will make a small variation of Example \ref{ex:exemHarris} in order to show that, even when $\overline{V}$ has no lightlike points, $\overline{M}$ could have them.
	Let 
	
	\[M= \left(\{(x,y)\in\R^2: 0< x< 1, -1\leq y\leq x\} \cup \left([1,2]\times [-1,1]\right)\right) \setminus \cup_{n} S_n
	\] be a manifold as in Figure \ref{fig:fig4} endowed with the induced Minkowski metric, where each $S_n$ is a spacelike segments obtained from a small variation of the lightlike segment joining $(1/n,1/n)$ and $(1,1-2/n)$. Due the fact that $(1/n,1/n)\ll (1,1-2/(n+1))$, such a variation can be taken in such a way that the past of the upper-right extreme of $S_{n+1}$ contains the down-left extreme of $S_n$. Let $V$ be the universal cover of $M$.
%
	
	The c-boundary (and so, the c-completion) of $M$ is represented on the right of Figure \ref{fig:fig4} and it is formed almost entirely by spacelike and timelike points. However, the points $(0,0)$ and $(1,1)$ are represented on the boundary by pairs of the form $(P_0,\emptyset)$ and $(P'_0,\emptyset)$ with $P_0\subset P'_0$\footnote{Observe that $(1/n,-1/n)\ll (1,1-2/(n+1))$, so it is possible to obtain timelike curves passing between $S_n$ and $S_{n+1}$.}, hence $M$ has lightlike boundary points. Topologically the c-completion $\overline{M}$ is Hausdorff, as it has the induced topology from $\R^2$.
	
	Now, if we look into the lifts of boundary points from $\overline{M}$ to $\overline{V}$, we observe that timelike and spacelike points are lifted to timelike and spacelike points respectively. However, there exist no lifts $(\lev{P}{},\emptyset)$ and $(\lev{P'}{},\emptyset)$ of $(P,\emptyset)$ and $(P',\emptyset)$ resp. such that $\lev{P}{}\subset \lev{P'}{}$, as any timelike curve moving from a point close to $(1,1)$ to a point close to $(0,0)$ should move between two segments $S_m$ and $S_{m+1}$, and so, it will move between different fibres of $V$ (recall again the behaviour of the universal covering described on Example \ref{ex:exemHarris}). Therefore, $\overline{V}$ will have no lightlike points. Finally, and due the fact that the topology around a point of $\overline{V}$ coincides again with the induced topology from $\R^2$, we have that $\overline{V}$ is also Hausdorff.  
      \end{exe}


\section{A physical application: Quotients on Robertson-Walker Spacetimes}\label{sec:application}
As a final section of this paper, we will show how our results are applicable to concrete and physically relevant models of spacetimes. Our main aim will be to apply Corollaries \ref{cor:aplicacion} and \ref{cor:expansiblyhomeomorphism} where, in addition to the finite chronology, we need to impose on $\overline{V}$ Hausdorffness and the non existence of lightlike points on $\overline{M}$ (recall also Lemma \ref{lem:lemnew}).

We will focus on the case of Robertson Walker models, even if our results are extensible to other more general ones (see Rem \ref{rem:aux4}). The c-completion of such a models is well known \cite[Section 4.2]{AF}, but we include here the details for completeness. Observe that we are not going to follow the original approach proposed in \cite{AF}, but the approach introduced in \cite[Section 3]{FHSIso2}. 

\smallskip 

Let $(\Sigma,\mathfrak{g}_{\Sigma})$ be a Riemannian manifold. Denote by $t:\R\times \Sigma\rightarrow \mathbb{R}$ and $\pi_{\Sigma}:\R\times \Sigma\rightarrow \Sigma$ the corresponding projections; and consider a smooth positive function $\alpha:\R\rightarrow (0,\infty)$. A Robertson Walker model with base $\Sigma$ and warping function $\alpha$ is given then by the pair $(V,\mathfrak{g})$, where

\begin{equation}\label{eq:RWModel}
V=\R\times \Sigma,\quad\hbox{and}\quad \mathfrak{g}=-dt^2 + (\alpha\circ t)\,\pi^*_{\Sigma}(\mathfrak{g}_{\Sigma}).
\end{equation} 
For simplicity, $\alpha\circ t$ will be denoted just by $\alpha(t)$ and, whenever there is no confusion, we will omit the pullback $\pi^*_{\Sigma}$. The chronological relation on these models is characterized as (see \cite[Prop. 3.1]{FHSIso2}):

\[
(t_0,x_0)\ll (t_1,x_1) \iff d(x_0,x_1)<\int_{t_0}^{t_1} \frac{1}{\sqrt{\alpha(s)}}ds
\]
where $d$ denotes the distance on $\Sigma$ defined by $\mathfrak{g}_{\Sigma}$. Thanks to previous characterization, it follows that any future terminal set $P$ is determined by the so-called Busemann functions. Such functions are defined in the following way: given a curve $c:[a,\Omega)\rightarrow \Sigma$ satisfying that $\mathfrak{g}_{\Sigma}(\dot{c},\dot{c})<1$, we define the associated Busemann function as:

\[
b_{c}(\cdot)=\lim_{t\rightarrow \Omega}\int_{0}^t \frac{1}{\sqrt{\alpha(s)}}ds - d(\cdot,c(t))
\]
Then, for any indecomposable past set $P$, it follows that $P=P(b_{c})$ for some curve $c$ with $\mathfrak{g}_{\Sigma}(\dot{c},\dot{c})<1$, where

\[
P(b_{c})=\{(t,x)\in V: t<b_{c}(x)\}
\]
(see \cite[Equation (3.3)]{FHSIso2}). If we have either $\Omega<\infty$; or $\Omega=\infty$ and $\int_{0}^{\infty}\frac{1}{\sqrt{\alpha(s)}}ds<\infty$; it follows that $c(t)\rightarrow x^*\in \Sigma_{C}$, where $\Sigma_{C}$ denotes the Cauchy completion associated to $(\Sigma,g_{\Sigma})$. Moreover, $b_{c}(\cdot)=d_{(\Omega,x^*)}(\cdot):=\int_{0}^\Omega \frac{1}{\sqrt{\alpha(s)}}ds - d(\cdot,x^*)$ (see \cite[Equations (3.7) and (3.8)]{FHSIso2}. In this way, and under the assumption of previous integral condition, we have that the future causal completion has the following point set structure:
\[
\hat{V}\equiv\Sigma_{C}\times \{\mathbb{R}\cup \{\infty\}\}
\] 
The study is completely analogous for the past orientation, where if we assume the integral condition $\int_{-\infty}^{0}\frac{1}{\sqrt{\alpha}(s)}ds<\infty$, the past causal completion is identified with:

\[
\check{V}\equiv\Sigma_{C}\times \{\mathbb{R}\cup \{-\infty\}\}.
\]

Finally, for the (total) c-completion, we only need to observe that past and future indecomposable past sets are $S$-related if they are associated to the same pair $(\Omega,x^*)\in \R\times \Sigma_{C}$ (see \cite[Equation (3.14)]{FHSIso2} and the paragraph above). In conclusion, the following result follows:

\begin{prop}\label{prop:propfinal}
	Let $(V,\mathfrak{g})$ be a Robertson Walker model as in \eqref{eq:RWModel}, and assume the following integral conditions
	\begin{equation}\label{eq:intcond} 
	\int_{0}^{\infty}\frac{1}{\sqrt{\alpha(s)}}ds<\infty,\quad\int_{-\infty}^{0}\frac{1}{\sqrt{\alpha(s)}}ds<\infty.
	\end{equation} 
	Then, the c-completion, as point set, becomes \[\overline{V}\equiv \Sigma_{C}\times \{\{-\infty\}\cup \mathbb{R}\cup \{\infty\}\}.\] Chronologically, the c-boundary has two copies, one for the future and one for the past, of the Cauchy completion $\Sigma_{C}$ formed by spacelike points; and timelike lines over each point of the Cauchy boundary of $\Sigma$. Topologically, and assuming that $\Sigma_{C}$ is locally compact, the chronological topology on $\overline{V}$ coincides with the product topology in $\Sigma_{C}\times \{\{-\infty\}\cup \mathbb{R}\cup \{\infty\}\}$ 
\end{prop}
\begin{proof}
	The pointset and causal structure can be deduced from previous comments (see also \cite[Theorem 4.2]{AF}). For the topological structure, we only need to recall \cite[Proposition 5.24]{FHSBuseman} is also applicable to this approach and, moreover, it is also true when $\Omega=\infty$ if the integral condition holds.
\end{proof}

Therefore, when the integral conditions are satisfied and the associated Cauchy completion $\Sigma_{C}$ is locally compact, $\overline{V}$ satisfies both, it is Hausdorff and has no lightlike points. Therefore, and as a consequence of Cor. \ref{cor:expansiblyhomeomorphism}:

\begin{thm}\label{thm:aux1} 
	Let $(V,\mathfrak{g})$ be a Robertson Walker model as in \eqref{eq:RWModel} and assume both, the integral conditions in \eqref{eq:intcond} and that $\Sigma_{C}$ is locally compact. Then, if $\pi:V\rightarrow M$ is a spacetime covering projection with associated group $G$, $(V,G)$ is finitely chronological and the $G$-orbits are closed for both $\hat{V}$ and $\check{V}$, then $\overline{V}/G$ and $\overline{M}$ are both, chronologically isomorphic and homeomorphic.
\end{thm}

\smallskip 

Obviously, our results are applicable in other Robertson Walker models without the integral conditions \eqref{eq:intcond}. For instance, the Anti-de Sitter model also satisfy both, it is Hausdorff and has no lightlike points (see \cite[Section 4.1]{AF}). Moreover, the only pairs in $\overline{V}$ with an empty component are of the form $(V,\emptyset)$ and $(\emptyset,V)$ (corresponding to $i^+$ and $i^-$, so, it follows readily that $\overline{M}$ has no lightlike points. Hence:

\begin{thm}\label{thm:final} 
Let $(V,\mathfrak{g})$ be the Anti-de Sitter model, that is, $V=\mathbb{R}\times (0,\infty)\times \mathbb{S}^2$ and 
	\[ \mathfrak{g}=-cosh^2(r)dt^2 + dr^2 + sinh^2(r)(d\theta^2+sin^2\theta d\phi^2).\]
Assume that we have a spacetime covering projection $\pi:V\rightarrow M$ with associated group $G$ in such a way that $(V,G)$ is finitely chronological. Then, $\overline{V}/G\equiv \overline{M}$.
\end{thm}

Previous result can be used, for instance, to calculate the c-completion of the BTZ blackhole models \cite{BTZ} and the Hawking-Page reference model \cite{HP}, which are obtained as suitable quotients of the 3-dimensional Anti-de Sitter model \cite{BHTZ,Wit1,Wit2}.

\begin{rem}\label{rem:aux4}
We would like to note finally that Theorem \ref{thm:aux1} is generalizable to other, more general, models of spacetimes. For instance, a similar result follows for Lorentz manifolds $(\R\times \Sigma,\mathfrak{g})$ with

\begin{equation}\label{eq:final} 
\mathfrak{g}=-dt^2 +   \sqrt{\alpha\circ t}\, \pi^*_{\Sigma}(\omega)\otimes dt +\sqrt{\alpha\circ t}\,dt \otimes\pi^*_{\Sigma}(\omega) + (\alpha\circ t)\,\pi^*_{\Sigma}(\mathfrak{g}_{\Sigma})
\end{equation}
where $\omega$ is a one-form of $\Sigma$.
Observe that such metrics are generalizations of Robetson-Walker models to the Standard Stationary settings. In fact, the theory developed in \cite{FHSBuseman} for the Stationary case is enough to study their c-completion (see \cite[Section 3]{FHSIso2}). 

The c-completion of the Standard Stationary case presents remarkable differences with respect to the Static one, mainly because its causality is no longer determined by a (regular) distance but by a (non-symmetric) generalized distance. That lack of symmetry is reflected on different structures for the future and past c-completions. For instance, future and past completions depends on different Cauchy completions (named by forward and backward Cauchy completions and denoted by $\Sigma_{C}^\pm$ resp., see \cite[Section 6]{FHSBuseman}). However, an under some mild hypotheses (the local compactness of $\Sigma_{C}^\pm$ and the well behaviour of the extended distance to such spaces, see \cite[Theorem 1.2]{FHSBuseman}), it follows analogous versions of Prop. \ref{prop:propfinal} and Theorem \ref{thm:aux1} for the model \eqref{eq:final}.
%
%
%
%
%
%
\end{rem}

\section*{Acknowledgments}

The authors are partially supported by the Spanish Grant MTM2013-47828-C2-2-P (MINECO and FEDER funds). L. Aké also acknowledges a grant funded by the Consejo Nacional de Ciencia y Tecnolog\'ia (CONACyT), M\'exico. JH would like to thank the Department of Algebra, Geometry and Topology of the University of M\'{a}laga for their kind hospitality while part of the work on this paper was being carried out.


\end{document}